\documentclass[dvipsnames]{amsart}

\usepackage{amsfonts, amssymb, amsmath, amsthm}
\usepackage{mathrsfs}                                                         
\usepackage{setspace}                                                         
\usepackage{geometry}                                                         
\usepackage{slashed}
\usepackage{xcolor}                                                           
\usepackage{graphicx}                                                         
\usepackage{comment}                                                          

\newtheorem{theorem}{Theorem}
\newtheorem{corollary}[theorem]{Corollary}

\newtheorem{proposition}[theorem]{Proposition}
\newtheorem{definition}[theorem]{Definition}
\newtheorem{remark}[theorem]{Remark}

\newtheorem{problem}[theorem]{Problem}

\numberwithin{equation}{section}
\numberwithin{theorem}{section}

\newcommand{\mf}[1]{\mathfrak{#1}}                                            
\newcommand{\mc}[1]{\mathcal{#1}}                                             
\newcommand{\ms}[1]{\mathsf{#1}}                                              
\newcommand{\mi}[1]{\mathscr{#1}}                                             

\newcommand{\R}{\mathbb{R}}                                                   
\newcommand{\Sph}{\mathbb{S}}                                                 

\newcommand{\gv}{\mathsf{g}}                                                  
\newcommand{\Lv}{\mathsf{m}}                                                  
\newcommand{\wv}{\mathsf{w}}                                                  
\newcommand{\Dv}{\mathsf{D}}                                                  

\newcommand{\gm}{\mf{g}}                                                      
\newcommand{\Dm}{\mf{D}}                                                      

\newcommand{\gb}[1]{\gm^{\scriptscriptstyle (#1)}}                            
\newcommand{\gbc}[1]{\bar{\gm}^{\scriptscriptstyle (#1)}}                     

\newcommand{\nablam}{\breve{\nabla}}                                          
\newcommand{\Boxm}{\breve{\Box}}                                              

\allowdisplaybreaks
\begin{document}

\title[Correspondence and Unique Continuation]{Bulk-boundary correspondences and unique continuation in asymptotically Anti-de Sitter spacetimes}

\author{Arick Shao}
\address{School of Mathematical Sciences\\
Queen Mary University of London\\
London E1 4NS\\
United Kingdom}
\email{a.shao@qmul.ac.uk}

\begin{abstract}
This article surveys the research presented by the author at the MATRIX Institute workshop \emph{Hyperbolic Differential Equations in Geometry and Physics} in April 2022.
The work is centred about establishing rigorous mathematical statements toward the AdS/CFT correspondence in theoretical physics, in particular in dynamical settings.
The contents are mainly based on the recent paper \cite{hol_shao:uc_ads_eve} with G.\ Holzegel that proved a unique continuation result for the Einstein-vacuum equations from asymptotically Anti-de Sitter (aAdS) conformal boundaries.
We also discuss some preceding results \cite{chatz_shao:uc_ads_gauge, hol_shao:uc_ads, hol_shao:uc_ads_ns, mcgill_shao:psc_aads, shao:aads_fg}, in particular novel Carleman estimates for wave equations on aAdS spacetimes, which laid the foundations toward the main results of \cite{hol_shao:uc_ads_eve}.
\end{abstract}

\maketitle

\section{Introduction} \label{sec.intro}

One of the most influential research directions in theoretical physics over the past three decades is the \emph{AdS/CFT correspondence}.
AdS/CFT---first formulated by Maldacena \cite{malda:ads_cft_0, malda:ads_cft}---roughly posits a \emph{correspondence between gravitational dynamics in asymptotically Anti-de Sitter spacetimes and conformal field theories on their boundaries}; see also the works of Gubser, Klebanov, and Polyakov \cite{gubs_kleban_polyak:ads_cft} and Witten \cite{witt:ads_cft} for further early developments.
AdS/CFT is also often associated with the term \emph{holography} \cite{thoo:holo_grav}, as the boundary theory is of one dimension lower than the gravitational theory.

In spite of extensive research and discussions in physics, surprisingly few mathematical statements pertaining to the AdS/CFT correspondence have been rigorously proven, or even formulated.
Furthermore, nearly all rigorous results have been within stationary backgrounds (see for instance \cite{alex_baleh_nach:inv_area, and_herz:uc_ricci, and_herz:uc_ricci_err, biq:uc_einstein, witt_yau:ads_cft}), whereas dynamical (time-dependent) settings are almost entirely untreated.

The works described in this article represent an ongoing programme to establish rigorous mathematical theory related to the AdS/CFT correspondence.
The aim is to rigorously state and prove statements along this direction, as well as to better understand the mechanisms behind such correspondences.
In particular, the key results here will be formulated in terms of unique continuation for the Einstein equations.
Also, although discussions of AdS/CFT in physics extend well beyond classical relativistic settings, here we will restrict our attention to within general relativity, as this provides us a physical theory with firmly established mathematical foundations.
 
\subsection{Anti-de Sitter Spacetime}

In Einstein's theory of general relativity, spacetime is modeled as an $( n + 1 )$-dimensional Lorentzian manifold $( \mi{M}, g )$; as a matter of convention, we assume $g$ has signature $( -, +, \dots, + )$.
In this article, we will only consider spacetimes that do not involve matter fields, so that $g$ satisfies the \emph{Einstein-vacuum equations} (\emph{EVE}),
\begin{equation}
\label{eq.eve_pre} \operatorname{Ric} [g] = \frac{ 2 \Lambda }{ n - 1 } \, g \text{,}
\end{equation}
with $\operatorname{Ric} [g]$ denoting the Ricci curvature associated to $g$, and $\Lambda \in \R$ being the \emph{cosmological constant}.
Furthermore, here we are concerned with the case in which $\Lambda < 0$.

\begin{remark}
While classical relativity lies in $3+1$ dimensions, here we work in any dimension, since many settings of interest within the AdS/CFT correspondence concern other dimensions.
\end{remark}

\emph{Anti-de Sitter}, or \emph{AdS}, \emph{spacetime} is the analogue of Minkowski spacetime in the case of negative cosmological constant.
More specifically, it is the maximally symmetric solution of the EVE
\begin{equation}
\label{eq.eve} \operatorname{Ric} [g] = -n \, g \text{,} \qquad \Lambda := \frac{ -n (n-1) }{2} < 0 \text{.}
\end{equation}
While there are many ways to describe AdS spacetime, one useful global representation is as the manifold $( \R_t \times \R^n_x, g_\textrm{AdS} )$, where the metric is given in polar coordinates by
\begin{equation}
\label{eq.ads} g_\textrm{AdS} := ( 1 + r^2 )^{-1} dr^2 - ( 1 + r^2 ) dt^2 + r^2 \mathring{\gamma} \text{,}
\end{equation}
with $\mathring{\gamma}$ being the unit round metric on $\Sph^{n-1}$.

\begin{remark}
One can rescale \eqref{eq.ads} to obtain a solution of \eqref{eq.eve_pre} for any $\Lambda < 0$.
The particular normalization $\smash{ \Lambda = \frac{ -n (n-1) }{2} }$ is convenient for simplifying various constants.
\end{remark}

For our purposes, it will be useful to express \eqref{eq.ads} using an inverted radius $\rho$, defined by
\begin{equation}
\label{eq.ads_rho} 4 r := \rho^{-1} ( 2 + \rho ) ( 2 - \rho ) \text{,} \qquad \rho \in ( 0, 2 ] \text{.}
\end{equation}
Observe that $\rho \searrow 0$ corresponds to ``infinity" $r \nearrow \infty$, while $\rho \nearrow 2$ corresponds to the centre $r \searrow 0$.
A direct computation shows that \eqref{eq.ads_rho} transforms \eqref{eq.ads} into the form
\begin{equation}
\label{eq.ads_fg} g_\textrm{AdS} = \frac{1}{ \rho^2 } \left[ d \rho^2 + ( - dt^2 + \mathring{\gamma}_{n-1} ) - \frac{1}{2} \rho^2 ( dt^2 + \mathring{\gamma}_{n-1} ) + \frac{1}{16} \rho^4 ( - dt^2 + \mathring{\gamma}_{n-1} ) \right] \text{,}
\end{equation}
with $\rho$ being a conformal boundary-defining function.
In anticipation of upcoming terminology, we will refer to \eqref{eq.ads_fg} as a \emph{Fefferman-Graham} (or \emph{FG}) \emph{gauge} for the AdS metric.

In particular, ignoring for the moment the conformal factor $\rho^{-2}$ in \eqref{eq.ads_fg}, we can then formally associate at ``$\rho = 0$" a timelike boundary for AdS spacetime:
\begin{equation}
\label{eq.ads_boundary} ( \mi{I}_\textrm{AdS}, \gm_\textrm{AdS} ) := ( \R_t \times \Sph^{n-1}, -dt^2 + \mathring{\gamma} ) \text{.}
\end{equation}
We refer to \eqref{eq.ads_boundary} as a \emph{conformal boundary}, or \emph{conformal infinity}, of AdS spacetime.

Similar to how one can study asymptotically flat spacetimes that only behave like Minkowski spacetime ``at infinity", we can similarly expand our outlook from AdS spacetime to \emph{asymptotically AdS} (\emph{aAdS}) \emph{spacetimes}.
Very roughly, aAdS spacetimes are those with similar qualitative properties ``at infinity" as AdS spacetime, in particular having a timelike conformal boundary.
We will provide a more precise description of aAdS spacetimes in the subsequent section.

\begin{remark}
A more well-known conformal transformation of AdS spacetime embeds AdS spacetime into half the Einstein cylinder, $\R_t \times \Sph^{n-1}_+$, with metric
\[
g_\textrm{AdS} = \frac{1}{ \cos^2 \theta } [ - dt^2 + d \theta^2 + ( \sin^2 \theta ) \, \mathring{\gamma} ] \text{.}
\]
However, the FG gauge will be more useful, as it extends directly to the study of aAdS metrics.
\end{remark}

\subsection{Unique Continuation}

Informally, our main objective is to address the following question, representing a purely classical statement in support of the AdS/CFT correspondence:

\begin{problem} \label{prb.correspondence_1}
Is there some one-to-one correspondence between (a) asymptotically AdS solutions $( \mi{M}, g )$ of the EVE \eqref{eq.eve} (representing ``gravitational dynamics"), and (b) appropriate data prescribed at the conformal boundary $\mi{I}$ (representing the ``conformal field theory")?
\end{problem}

If we lack additional foresight from partial differential equations (PDEs), then a reasonable attempt at Problem \ref{prb.correspondence_1} might be to \emph{solve \eqref{eq.eve} given appropriate Cauchy (i.e.\ Dirichlet and Neumann) data on $\mi{I}$}.
However, this approach immediately runs into fundamental issues, since the EVE is primarily hyperbolic (in the appropriate gauges).
As a result, \eqref{eq.eve} becomes ill-posed when data is prescribed on a timelike hypersurface, such as the conformal boundary $( \mi{I}, \gm )$.
In particular, given Cauchy data on $\mi{I}$, a solution to \eqref{eq.eve} may not exist, much less depend continuously on the data.

\begin{remark}
In contrast, the appropriate well-posed problem for the EVE in aAdS settings is the initial-boundary value problem, in which Cauchy data is prescribed on a spacelike hypersurface, and \emph{only one} piece of boundary data (e.g.\ Dirichlet, Neumann, or Robin boundary condition) is imposed on the conformal boundary.
For further discussions, see \cite{enc_kam:aads_wave, enc_kam:aads_hol, frie:aads_conf, vasy:wave_aads, warn:wave_aads}.
\end{remark}

We are thus forced to attack Problem \ref{prb.correspondence_1} differently.
The key idea is that we instead formulate Problem \ref{prb.correspondence_1} as a \emph{unique continuation} problem---this is a classical PDE question for ill-posed settings that has an extensive history and literature; see, e.g., \cite{cald:unique_cauchy, carl:uc_strong, holmg:uc_anal, hor:lpdo4} for classical results.
In short, for unique continuation problems, one avoids solving the PDE altogether and instead asks:\ \emph{if the solution exists, then must it be uniquely determined by the prescribed data}?

In light of this, we can restate Problem \ref{prb.correspondence_1} as a unique continuation problem for the EVE:

\begin{problem} \label{prb.correspondence_2}
If two aAdS spacetimes $( \mi{M}, g )$ and $( \bar{\mi{M}}, \bar{g} )$ solving the Einstein-vacuum equations \eqref{eq.eve} have the same conformal boundary data, then must $g$ and $\bar{g}$ be isometric?
\end{problem}

Very informally, the answer to Problem \ref{prb.correspondence_2}---the main result of the joint work \cite{hol_shao:uc_ads_eve} with Holzegel---is \emph{``yes, assuming certain conditions on the conformal boundary $( \mi{I}, \gm )$ hold"}.
The goal of the upcoming sections is to work toward a precise formulation of Problem \ref{prb.correspondence_2} and a precise statement of the main result of \cite{hol_shao:uc_ads_eve}.
Moreover, we will discuss the main ideas behind the proof of this result.

\subsection{Organisation of the Article}

In Section \ref{sec.aads}, we give a precise description of aAdS spacetimes, the setting of our analysis, and we introduce the notion of vertical tensor fields, the main quantities that we will work with.
Next, Section \ref{sec.main} is dedicated to the precise statements of our main result (see Theorem \ref{thm.correspondence}) and its corollaries, as well as the geometric assumptions required for these results to hold.
In Section \ref{sec.carleman}, we discuss recently established Carleman estimates for tensorial wave equations in aAdS spacetimes, which is the main technical innovation behind our results.
Section \ref{sec.proof} introduces some of the main ideas behind the proof of our main result.
Finally, in Section \ref{sec.open}, we briefly discuss some additional questions of interest that are related to the contents of this article.

\subsection{Acknowledgments}

The author thanks the MATRIX Institute for holding the workshop, which provided a wonderful environment for discussions.
Special thanks go to Jesse Gell-Redman, Andrew Hassell, Todd Oliynik, and Volker Schlue for their efforts organising the workshop, as well as to Sakis Chatzikaleas, Gustav Holzegel, Alex McGill, and Simon Guisset for the collaborations and discussions that made this research possible.
Furthermore, a considerable portion of this work was supported by EPSRC grant EP/R011982/1, \emph{Unique continuation for geometric wave equations, and applications to relativity, holography, and controllability}.

\section{Asymptotically AdS Spacetimes} \label{sec.aads}

While we had previously characterized aAdS spacetimes as ``roughly like AdS near and at the conformal boundary", we will need a much more precise definition if we are to state precise results.
In this section, we provide the mathematical setup for studying various relevant quantities on aAdS backgrounds, both near and at the conformal boundary.

\subsection{The Fefferman-Graham Gauge}

Similar to the AdS metric, one can apply a similar transformation to metrics that are only asymptotically like AdS into a FG gauge, characterized by a boundary defining function $\rho$ that is both normalised and fully decoupled from the other components.
As a result, we will, as a matter of convenience, \emph{define} the aAdS spacetimes that we consider in terms of such FG gauges.
We refer to these backgrounds as \emph{FG-aAdS segments}, representing an appropriate portion of an aAdS spacetime near the conformal boundary.

\begin{remark}
See \cite{gra:vol_renorm} for a treatment of FG gauges in asymptotically hyperbolic manifolds, that is, the Riemannian analogue of our setting.
In fact, the argument in \cite{gra:vol_renorm} for reducing more general metrics into FG gauges extends directly to Lorentzian, aAdS settings.
\end{remark}

\begin{definition} \label{def.aads}
Let $( \mi{I}, \gm )$ be a $n$-dimensional Lorentzian manifold, and let $\rho_0 > 0$.
We say that $( \mi{M} := ( 0, \rho_0 ]_\rho \times \mi{I}, g )$ is an \emph{FG-aAdS segment}, with \emph{conformal infinity} $( \mi{I}, \gm )$, if $g$ has the form
\begin{equation}
\label{eq.aads} g = \rho^{-2} ( d \rho^2 + \gv ) \text{,}
\end{equation}
where $\gv := \gv ( \rho )$, $\rho \in ( 0, \rho_0 ]$, is a smooth family of Lorentzian metrics on $\mi{I}$ such that $\gv$ has a continuous limit $\gv (0) := \gm$ as $\rho \searrow 0$.
Furthermore:
\begin{itemize}
\item We refer to the form \eqref{eq.aads} as the \emph{Fefferman-Graham} (\emph{FG}) \emph{gauge condition}.

\item We say $( \mi{M}, g )$ is \emph{vacuum} iff $g$ also satisifes the EVE \eqref{eq.eve}.
\end{itemize}
\end{definition}

Observe in particular that AdS spacetime can itself be formulated as a vacuum FG-aAdS segment, with the standard conformal infinity \eqref{eq.ads_boundary}.
In particular, the FG gauge \eqref{eq.ads_fg} for the AdS metric is precisely of the form \eqref{eq.aads}.
Similarly, the usual Schwarzschild-AdS and Kerr-AdS spacetimes can also be expressed as vacuum FG-aAdS segments, at least near the conformal boundary.
More generally, a large class of vacuum FG-aAdS segments with conformal infinity \eqref{eq.ads_boundary} can be constructed by solving initial-boundary value problems for the EVE; see \cite{enc_kam:aads_hol, frie:aads_conf}.

\begin{remark}
Note that Definition \ref{def.aads} allows for very general conformal boundary topologies.
In particular, we can consider settings with flat conformal boundaries,
\begin{equation}
\label{eq.flat_boundary} \mi{I}_\textrm{pAdS} := \R_t \times \R^n \text{,} \qquad \mi{I}_{tAdS} := \R_t \times \mathbb{T}^n \text{.}
\end{equation}
These are realized as conformal boundaries for the \emph{planar} and \emph{toric AdS spacetimes}.
\end{remark}

\begin{remark}
We refer to $\gv$, which fully describes the spacetime geometry, as the \emph{vertical metric}.
We will discuss vertical tensor fields more generally at the end of this section.
\end{remark}

\subsection{Fefferman-Graham Expansions}

We now turn our attention to vacuum FG-aAdS segments, and can immediately ask the following:\ \emph{if $( \mi{M}, g )$ is also vacuum, then what structure does this impose on $g$ at the conformal boundary}?
The answer comes from the seminal work of Fefferman and Graham \cite{fef_gra:conf_inv, fef_gra:amb_met} on the ambient construction.
While these concerned extending vacuum metrics from null cones, this has since been widely adapted to aAdS settings in the physics literature.

In short, when $( \mi{M}, g )$ is vacuum, one can derive a formal series expansion for $g$ near $\rho=0$:
\begin{align}
\label{eq.aads_exp} \gv ( \rho ) &= \begin{cases} \gb{0} + \gb{2} \rho^2 + \dots + \gb{n-1} \rho^{n-1} + \gb{n} \rho^n + \dots & n \text{ odd,} \\ \gb{0} + \gb{2} \rho^2 + \dots + \gb{n-2} \rho^{n-2} + \gb{\star} \rho^n \log \rho + \gb{n} \rho^n + \dots & n \text{ even.} \end{cases}
\end{align}
The above is commonly known as the \emph{Fefferman-Graham expansion} for $g$.
The $\gb{k}$'s and $\gb{\star}$ are covariant symmetric rank-$2$ tensor fields on $\mi{I}$, with the anomalous logarithmic coefficient $\gb{\star}$ only appearing when the boundary dimension $n$ is even.
Note the leading coefficient $\gb{0} = \gm$ is simply the conformal boundary metric.
Moreover, \eqref{eq.eve} implies that \emph{all the coefficients $\gb{k}$ for $0 < k < n$---as well as $\gb{\star}$ when $n$ is even---are determined locally by $\gm$ and its derivatives}.

\begin{remark}
We will use the notations $\gm$ and $\gb{0}$ interchangeably, depending on context.
\end{remark}

The situation changes once we reach $\gb{n}$.
Here, one can derive from \eqref{eq.eve} that both the divergence and the trace of $\gb{n}$ are determined by $\gm$ and its derivatives.
On the other hand, the remaining components of $\gb{n}$ are free, meaning that they are not formally determined by the EVE.

Finally, the formal expansion \eqref{eq.aads_exp} can be further continued beyond $\gb{n}$, with all subsequent coefficients formally determined by the pair $( \gb{0}, \gb{n} )$ alone.
Also, when $n$ is even, the expansion remains polyhomogeneous beyond $\gb{n}$, involving additional powers of $\log \rho$.

Now, when $( \gb{0}, \gb{n} )$ is real-analytic, Kichenassamy \cite{kichen:fg_log} showed, via Fuchsian techniques, that \eqref{eq.aads_exp} converges near $\rho = 0$ to a vacuum aAdS metric.
However, analyticity is a far too restrictive condition, as this is not satisfied by generic aAdS metrics (e.g.\ arising from an initial-boundary value problem), for which the expansion \eqref{eq.aads_exp} is merely formal.

In these non-analytic settings, in which \eqref{eq.aads_exp} fails to converge, one still expects that a \emph{partial FG expansion} holds, up to some finite order.
(An apt analogy would be applying Taylor's theorem as opposed to Taylor series.)
This was rigorously proved by the author in \cite{shao:aads_fg}---more specifically, assuming minimal regularity conditions, $\gv$ retains the form \eqref{eq.aads_exp}, but only up to $n$-th order.
We state a simplified version of this result below (see \cite[Theorem 3.3]{shao:aads_fg} for the full statement):

\begin{theorem}[\cite{shao:aads_fg}] \label{thm.aads_fg}
Let $( \mi{M}, g )$ be a vacuum FG-aAdS segment, with $n > 1$.
Moreover, suppose:
\begin{itemize}
\item $\gv$ is locally bounded in $C^{ n + 4 }$ in the $\mi{I}$-directions, up to the conformal boundary.

\item $\mi{L}_\rho \gv$ is locally bounded in $C^0$ in the $\mi{I}$-directions, up to the conformal boundary.
\end{itemize}
Then, $\gv$ satisfies the following partial FG expansion near $\rho = 0$:
\begin{equation}
\label{eq.aads_fg} \gv = \begin{cases} \sum_{ k = 0 }^{ \frac{ n - 1 }{2} } \gb{ 2k } \rho^{ 2 k } + \gb{n} \rho^n + \ms{o} ( \rho^n ) & \text{$n$ odd,} \\ \sum_{ k = 0 }^{ \frac{ n - 2 }{2} } \gb{ 2k } \rho^{ 2 k } + \gb{\star} \rho^n \log \rho + \gb{n} \rho^n + \ms{o} ( \rho^n ) & \text{$n$ even.} \end{cases}
\end{equation}
Furthermore, the coefficients of the expansion satisfies the following:
\begin{itemize}
\item $\gb{2k}$, for any $0 \leq 2 k < n$, is determined by $\gm$ and its derivatives:
\begin{equation}
\label{eq.fg_low} \gb{2k} = \mc{R}_n^{ (2k) } ( \partial^{ \leq 2k } \gm ) \text{.}
\end{equation}

\item In particular, $\gb{0} = \gm$, and $-\gb{2}$ is the Schouten tensor for $\gm$ whenever $n \geq 3$:
\begin{equation}
\label{eq.fg_schouten} -\gb{2} = \mc{P} [ \gm ] := \frac{1}{n-2} \left( \operatorname{Ric} [\gm] - \frac{1}{2(n-1)} \operatorname{R} [\gm] \cdot \gm \right) \text{.}
\end{equation}

\item When $n$ is even, $\gb{\star}$ is also determined by $\gm$ and its derivatives:
\begin{equation}
\label{eq.fg_high} \gb{\star} = \mc{R}_n^{ (\star) } ( \partial^{ \leq n } \gm ) \text{.}
\end{equation}

\item Both the $\gm$-divergence and $\gm$-trace of $\gb{n}$ are determined by $\gm$ and its derivatives:
\begin{equation}
\label{eq.fg_constraint} \operatorname{div}_{ \gm } \gb{n} = \mc{R}_n^{ \operatorname{div} } ( \partial^{ \leq n+1 } \gm ) \text{,} \qquad \operatorname{tr}_{ \gm } \gb{n} = \mc{R}_n^{ \operatorname{tr} } ( \partial^{ \leq n } \gm ) \text{.}
\end{equation}
In fact, when $n$ is odd, both $\mc{R}_n^{ \operatorname{div} }$ and $\mc{R}_n^{ \operatorname{tr} }$ are identically zero.
\end{itemize}
\end{theorem}

\begin{remark}
Here and throughout the rest of this article, we use $\mi{L}$ to denote the Lie derivative, and $\mi{L}_\rho$ for the Lie derivative in the ``$( 0, \rho_0 ]$-direction" of $\mi{M}$.
\end{remark}

\begin{remark}
The relations $\mc{R}_n^{ (2k) }$, $\mc{R}_n^{ (\star) }$, $\mc{R}_n^{ \operatorname{div} }$, $\mc{R}_n^{ \operatorname{tr} }$ in \eqref{eq.fg_low}--\eqref{eq.fg_constraint} are universal, in that they only depend on the dimension $n$, and not on $\mi{I}$ or $\gm$.
Moreover, each of these relations could in principle be computed, although the exact formulas tend to be exceedingly complicated.
\end{remark}

\begin{remark}
The partial FG expansion can be continued beyond $\gb{n}$ to any arbitrary (but finite) order, as long as one assumes sufficient additional regularity on $\gv$.
\end{remark}

The coefficients $\gb{2k}$ ($0 < 2k < n$) and $\gb{\star}$ are generated by computing the limits of $\rho$-derivatives of $\gv$ at the conformal boundary.
These boundary limits are obtained through an inductive process by taking successive $\rho$-derivatives of (an appropriate form of) the EVE.
This induction continues until one reaches $\gb{n}$, which---other than \eqref{eq.fg_constraint}---cannot be obtained from the EVE.
Nonetheless, one can still extract existence of the free coefficient $\gb{n}$ in \eqref{eq.aads_fg}, although it is not determined by any of the preceding coefficients.
For details, see the discussions in \cite{shao:aads_fg}.

\begin{remark}
In particular, the conclusions of \cite{shao:aads_fg} show that in practice, one needs not separately assume that $\gb{n}$ and the $\rho^n$-term in the partial FG expansion exists.
\end{remark}

In particular, Theorem \ref{thm.aads_fg} shows that even though we no longer have a full infinite series expansion, \emph{we can nonetheless still view $( \gb{0}, \gb{n} )$ as free boundary data} (excepting the constraints \eqref{eq.fg_constraint}) \emph{for the EVE}.
Additionally, in the physics literature, $\gb{n}$ is closely connected to the \emph{stress-energy tensor} for the boundary conformal field theory; for further discussions, see \cite{deharo_sken_solod:holog_adscft, sken:aads_set}.

For future convenience, we introduce the following terminology for our boundary data:

\begin{definition} \label{def.holo_data}
In general, we refer to a triple $( \mi{I}, \gb{0}, \gb{n} )$ as \emph{holographic data} iff $( \mi{I}, \gb{0} )$ is an $n$-dimensional Lorentzian manifold, and $\gb{n}$ is a symmetric $2$-tensor on $\mi{I}$ satisfying \eqref{eq.fg_constraint}.

Furthermore, given a vacuum FG-aAdS segment $( \mi{M}, g )$, we refer to the $( \mi{I}, \gb{0}, \gb{n} )$ obtained from Theorem \ref{thm.aads_fg} as the \emph{holographic data} associated to $( \mi{M}, g )$.
\end{definition}

Finally, one can explicitly compute the partial FG coefficients for the Schwarzschild-AdS metric with mass $M \in \R$.
These are well-known in the physics literature (see, for instance, \cite{bautier_englert_rooman_spindel:fg_ads, tetradis:bh_holog}), but they are also computed in \cite{shao:aads_fg}.
In particular, when $n > 2$ and $n \neq 4$, we have
\begin{equation}
\label{eq.fg_schw} \gb{0} = - dt^2 + \mathring{\gamma} \text{,} \qquad \gb{2} = - \frac{1}{2} ( dt^2 + \mathring{\gamma} ) \text{,} \qquad \gb{n} = \frac{M}{n} [ ( n - 1 ) dt^2 + \mathring{\gamma} ] \text{.}
\end{equation}
Note $\gb{0}$ is the same for all Schwarzschild-AdS metrics, and the mass $M$ is encoded in $\gb{n}$.

\subsection{Gauge Transformations}

Though the partial FG expansions of Theorem \ref{thm.aads_fg} capture near-boundary geometries of vacuum FG-aAdS segments, they fail to account for one additional gauge freedom in our setting.
The issue is that Definition \ref{def.aads} of FG-aAdS segments presupposes a choice of FG gauge.
However, there does exist a nontrivial class of transformations $\rho \mapsto \bar{\rho}$ of the boundary defining function that preserve the form \eqref{eq.aads} of the FG gauge.

The key observation is that \eqref{eq.aads} can be achieved by any function $\bar{\rho}$ satisfying
\begin{equation}
\label{eq.fg_define} \bar{\rho} > 0 \text{,} \qquad \lim_{ \rho \searrow 0 } \bar{\rho} = 0 \text{,} \qquad \bar{\rho}^{-2} \, g ( d \bar{\rho}, d \bar{\rho} ) = 1 \text{.}
\end{equation}
(The critical condition is that the $\bar{\rho}^{-2} g$-gradient of $\bar{\rho}$ is unit, while the other components are then transported along the integral curves of this gradient; see \cite{gra:vol_renorm} for details.)
Adopting the ansatz
\[
\bar{\rho} = e^\ms{a} \rho \text{,}
\]
then \eqref{eq.fg_define} reduces to fully nonlinear initial-value problem
\[
( 2 + \rho \mi{L}_\rho \ms{a} ) \mi{L}_\rho \ms{a} + \rho \, \gv^{-1} ( \Dv \ms{a}, \Dv \ms{a} ) = 0 \text{,}
\]
which can be uniquely solved from $\rho = 0$ using the method of characteristics.
This implies that \emph{for any $\mf{a} \in C^\infty ( \mi{I} )$, one can find a unique FG gauge, with boundary defining function $\bar{\rho}$, such that}
\begin{equation}
\label{eq.fg_rho} \lim_{ \rho \searrow 0 } \frac{ \bar{\rho} }{ \rho } = e^\mf{a} \text{.}
\end{equation}

\begin{remark}
Using an appropriate isometry, one can then formulate the FG gauge defined with respect to $\bar{\rho}$ as an FG-aAdS segment (at least near the conformal boundary).
\end{remark}

Of course, these different FG gauges represent the same physical object, since they are different representations of the same spacetime.
On the other hand, a change of FG gauge $\rho \mapsto \bar{\rho}$ induces a corresponding transformation of the partial FG expansion.
For instance, one can show that for the FG change satisfying \eqref{eq.fg_rho}, the boundary metric $\gb{0}$ undergoes a conformal transformation:
\begin{equation}
\label{eq.fg_conformal} \gb{0} \mapsto \gbc{0} = e^{ 2 \mf{a} } \gb{0} \text{.}
\end{equation}
In other words, only the conformal class $[ \gm ]$ of the boundary metric can be invariantly associated with a given aAdS spacetime; this motivates the term ``conformal boundary" for $( \mi{I}, \gm )$.

The other coefficients in \eqref{eq.aads_fg} are also transformed via changes of FG gauge (see \cite{deharo_sken_solod:holog_adscft, imbim_schwim_theis_yanki:diffeo_holog}), though the explicit formulas quickly become complicated.
In general, there do exist explicitly computable (given enough effort) universal functions such that the FG coefficients transform as
\begin{align}
\label{eq.fg_change} \gb{2k} &\mapsto \gbc{2k} = \mc{G}_n^{2k} ( \partial^{ \leq 2k } \mf{a}, \partial^{ \leq 2k } \gb{0} ) \text{,} \\
\notag \gb{\star} &\mapsto \gbc{\star} = \mc{G}_n^{\star} ( \partial^{ \leq n } \mf{a}, \partial^{ \leq n } \gb{0} ) \text{,} \\
\notag \gb{n} &\mapsto \gbc{n} = \mc{G}_n^n ( \partial^{ \leq n } \mf{a}, \partial^{ \leq n } \gb{0}, \gb{n} ) \text{,}
\end{align}
where $0 \leq 2 k < n$.
(In particular, $-\gb{2}$ transforms like the Schouten tensor when $n \geq 3$.)

The physical significance, then, is that \emph{pairs $( \gb{0}, \gb{n} )$ and $( \gbc{0}, \gbc{n} )$ that are related via \eqref{eq.fg_conformal} and \eqref{eq.fg_change} should be viewed as ``the same"}, since they arise from the same aAdS spacetime:

\begin{definition} \label{def.gauge_equiv}
Let $( \mi{I}, \gb{0}, \gb{n} )$, $( \bar{\mi{I}}, \gbc{0}, \gbc{n} )$ denote holographic data, and let $\phi: \mi{I} \leftrightarrow \bar{\mi{I}}$ be a diffeomorphism.
We say that $( \mi{I}, \gb{0}, \gb{n} )$ and $( \bar{\mi{I}}, \gbc{0}, \gbc{n} )$ are $\phi$\emph{-gauge-equivalent} on $\mi{D} \subset \mi{I}$ iff there exists $\mf{a} \in C^\infty ( \mi{I} )$ such that the following relations hold:
\begin{equation}
\label{eq.gauge_equiv} \phi_\ast \gbc{0} |_{ \mi{D} } = e^{ 2 \mf{a} } \gb{0} |_{ \mi{D} } \text{,} \qquad \phi_\ast \gbc{n} |_{ \mi{D} } = \mc{G}^n_n ( \partial^{ \leq n } \mf{a}, \partial^{\leq n } \gb{0}, \gb{n} ) |_{ \mi{D} } \text{.}
\end{equation}
\end{definition}

\begin{remark}
As usual, $\phi_\ast$ in \eqref{eq.gauge_equiv} denotes the pullback of $\phi$.
In practice, by pulling back through $\phi$, we can always restrict to the case in which $\bar{\mi{I}} = \mi{I}$ and $\phi$ is the identity map.
\end{remark}

\subsection{The Vertical Tensor Calculus}

This subsection is devoted to a more detailed discussion of the main quantities we will analyze---vertical tensor fields.
We remark that readers who are only interested in the statement of the main result and its corollaries can skip this part entirely.
On the other hand, the contents here will be relevant to the proof of the main result.

We begin with a precise definition of vertical tensor fields:

\begin{definition} \label{def.vertical}
Let $( \mi{M}, g )$ be an FG-aAdS region.
Then:
\begin{itemize}
\item The \emph{vertical bundle} $\ms{V}^k_l \mi{M}$ of rank $( k, l )$ over $\mi{M}$ is defined to be the manifold consisting of all tensors of rank $( k, l )$ on each level set of $\rho$ in $\mi{M}$:
\begin{equation}
\label{eq.vertical} \ms{V}^k_l \mi{M} = \bigcup_{ \sigma \in ( 0, \rho_0 ] } T^k_l \{ \rho = \sigma \} \text{.}
\end{equation}

\item A section of $\ms{V}^k_l \mi{M}$ is called a \emph{vertical tensor field} of rank $( k, l )$.
\end{itemize}
\end{definition}

In short, a vertical tensor field on $\mi{M}$ is a tensor field that is entirely tangent in every component to the level sets of $\rho$.
Observe also that vertical tensor fields can be equivalently characterized as a family, parametrized by $\rho$, of tensor fields on $\mi{I}$.
In particular, the second interpretation provides a natural notion of limits at the conformal boundary---for a vertical tensor field $\ms{A} = \ms{A} ( \rho )$, its limit as $\rho \searrow 0$ (when exists) is simply a tensor field of the same rank at $\mi{I}$.
This allows one to easily connect quantities in the bulk spacetime with those on the conformal boundary.

One basic example of a vertical tensor field is the vertical metric $\gv = \gv ( \rho )$ in \eqref{eq.aads}.
Note that the boundary limit of $\gv$ is the conformal boundary metric $\gm = \gb{0}$.
Moreover, $\gv$ induces a notion of \emph{vertical covariant derivative} $\Dv$; more specifically, since $\gv ( \rho )$ is a Lorentzian metric for each $\rho$, we can define $\Dv$ by aggregating the Levi-Civita connections of all the $\gv ( \rho )$'s.

Analogues of vertical tensor fields have been widely used in mathematical relativity; these have sometimes been referred to as horizontal tensor fields.
(Common examples include the connection and curvature components on level spheres in a double null foliation.)
However, we will also extend these ideas in some novel ways in order to reach our main results.

\begin{remark}
We adopt the term ``vertical" rather than ``horizontal" here, since the hypersurfaces in our foliation---the level sets of $\rho$---are timelike rather than spacelike.
\end{remark}

In particular, since tensorial wave equations play a key role in the proofs of our results, we will want to make sense of a \emph{spacetime wave operator $\Boxm_g$ applied to vertical tensor fields}.
Furthermore, we aim to do this in a covariant manner, so that the usual operations of geometric analysis---e.g.\ Leibniz rules, integrations by parts---continue to hold.
The upshot is that we will be able to treat vertical tensor fields covariantly, in nearly the same manner as we do scalar fields.

We achieve these goals in a series of steps, with the first being to make sense of first spacetime covariant derivatives of vertical tensor fields.
While the vertical connection $\Dv$ already makes sense of derivatives in vertical directions, we now want to also extend $\Dv$ to the $\rho$-direction.
There are many ways to accomplish this, the standard approach being to define $\Dv_\rho$ as an appropriate projection of the $g$-covariant derivative.
However, here we elect to do this differently, by instead projecting the $\rho^2 g$-covariant derivative.
This is primarily a matter of convenience, since $\rho^2 g$ is well-defined at the conformal boundary, so this definition yield formulas with fewer singular weights.
Observe that, crucially, this extended covariant derivative remains torsion-free:
\[
\Dv_\rho \gv = 0 \text{.}
\]

However, the key issues arise from making proper sense of \emph{second spacetime} derivatives of vertical tensor fields.
To do this covariantly, one requires a small detour---we must extend our vertical tensor calculus to \emph{mixed tensor fields}, containing both spacetime and vertical components:

\begin{definition} \label{def.mixed}
Let $( \mi{M}, g )$ be an FG-aAdS segment.
Then:
\begin{itemize}
\item We define the \emph{mixed bundle} of ranks $( \kappa, \lambda; k, l )$ over $\mi{M}$ to be the tensor product bundle
\begin{equation}
\label{eq.mixed} T^\kappa_\lambda \ms{V}^k_l \mi{M} := T^\kappa_\lambda \mi{M} \otimes \ms{V}^k_l \mi{M} \text{.}
\end{equation}

\item We refer to sections of $T^\kappa_\lambda \ms{V}^k_l \mi{M}$ as \emph{mixed tensor fields} of ranks $( \kappa, \lambda; k, l )$.
\end{itemize}
\end{definition}

The tensor product structure of the mixed bundles induces a natural notion of covariant derivative $\smash{\nablam}$ for mixed tensor fields---in short, \emph{$\smash{\nablam}$ acts like the spacetime connection $\nabla$ on spacetime components and like the vertical connection $\Dv$ on vertical components}.
This behavior can be generated from the following defining relation, which can be shown to uniquely define $\smash{\nablam}$:\ given any vector field $X$ on $\mi{M}$, tensor field $G$ on $\mi{M}$, and vertical tensor field $\ms{B}$, we set
\[
\nablam_X ( G \otimes \ms{B} ) := \nabla_X G \otimes \ms{B} + G \otimes \Dv_X \ms{B} \text{.}
\]

The torsion-free property extends also to the mixed connection $\smash{\nablam}$.
As both the spacetime metric $g$ and the vertical metric $\gv$ are mixed tensor fields, our definition of $\smash{\nablam}$ immediately implies
\[
\nablam g = 0 \text{,} \qquad \nablam \gv = 0 \text{.}
\]

Mixed bundles allow us to now covariant define a spacetime Hessian of a vertical tensor field $\ms{A}$, say of rank $( k, l )$.
Recall that the first derivative $\smash{\nablam \ms{A}}$ is simply the extended vertical connection applied to $\ms{A}$.
The trick is now to view $\smash{\nablam \ms{A}}$ as a mixed tensor field of rank $( 0, 1; k, l )$.
The Hessian \emph{$\smash{\nablam^2 \ms{A}}$ is then naturally defined as $\smash{\nablam}$ applied to the mixed tensor field $\smash{\nablam \ms{A}}$}.

Finally, for the wave operator, since $\smash{\nablam^2 \ms{A}}$ is a mixed tensor field of rank $( 0, 2; k, l )$, we can simply define $\smash{\Boxm_g \ms{A}}$ as the $g$-trace of $\smash{\nablam^2 \ms{A}}$ along the two spacetime components.

\begin{remark}
Since all our connections, by their construction, satisfy the requisite Leibniz rules and are torsion-free, then, with regards to integration by parts, we can treat both vertical and mixed tensor fields in the same way that we treat scalar fields and spacetime tensor fields.
\end{remark}

\begin{remark}
Mixed tensor fields and extended wave operators $\smash{\Boxm_g}$ originated from \cite{shao:ksp}; for aAdS contexts, an earlier formulation using horizontal tensor fields was used \cite{hol_shao:uc_ads, hol_shao:uc_ads_ns}.
(Similar notions were also independently developed and used in \cite{keir:weak_null}.)
The full vertical (and mixed) tensor calculus, in the form shown here, was first constructed in \cite{mcgill_shao:psc_aads, shao:aads_fg} and then used in \cite{chatz_shao:uc_ads_gauge, hol_shao:uc_ads_eve}.
\end{remark}

\section{The Main Result} \label{sec.main}

Having provided, in Section \ref{sec.aads}, a more complete description of our aAdS spacetimes and their properties, we can now return to the main correspondence question posed in this article.
In particular, we can now give a more precise formulation of Problem \ref{prb.correspondence_2}:

\begin{problem}[Correspondence] \label{prb.correspondence}
Let $( \mi{M}, g )$ and $( \bar{\mi{M}}, \bar{g} )$ be two vacuum FG-aAdS segments, and suppose that their respective associated holographic data, $( \mi{I}, \gb{0}, \gb{n} )$ and $( \bar{\mi{I}}, \gbc{0}, \gbc{n} )$, are gauge-equivalent.
Then, must $g$ and $\bar{g}$ be isometric to each other?
\end{problem}

Compared to Problem \ref{prb.correspondence_2}, we have replaced the vague ``aAdS spacetimes" with the precisely defined FG-aAdS segments of Definition \ref{def.aads}.
Furthermore, the ``conformal boundary data" mentioned in Problem \ref{prb.correspondence_2} have been clarified as holographic data, up to gauge-equivalence.

A closely related question is that of \emph{symmetry extension}---\emph{are (appropriately defined) symmetries of the holographic data necessarily inherited by the bulk aAdS spacetime}?
Using the terminology of Section \ref{sec.aads}, we can again provide a more precise statement:

\begin{problem}[Symmetry Extension] \label{prb.symmetry}
Let $( \mi{M}, g )$ be a vacuum FG-aAdS segment, with holographic data $( \mi{I}, \gb{0}, \gb{n} )$.
If $\psi$ is a symmetry of $( \mi{I}, \gb{0}, \gb{n} )$ (that is, a diffeomorphism $\psi: \mi{I} \leftrightarrow \mi{I}$ that preserves $\gb{0}$ and $\gb{n}$ up to gauge-equivalence), then must $\psi$ extend to a isometry of $g$?
\end{problem}

Our main result, Theorem \ref{thm.correspondence} below, will give a positive resolution to Problem \ref{prb.correspondence}.
Answers to Problem \ref{prb.symmetry} will then follow as corollaries; see Theorems \ref{thm.symmetry} and \ref{thm.killing}.

\subsection{Previous Literature}

Before stating our main theorems, we first review some of the existing literature related to Problems \ref{prb.correspondence} and \ref{prb.symmetry}, starting with stationary settings.

First, Biquard \cite{biq:uc_einstein} positively answered the Riemannian analogue of Problem \ref{prb.correspondence}.
More specifically, \cite{biq:uc_einstein} considered asymptotically hyperbolic Einstein manifolds $( N, h )$, which have both a conformal boundary $( \partial N, \mf{h} )$ and an analogous Fefferman-Graham expansion from $\partial N$.
The main result of \cite{biq:uc_einstein} then showed that \emph{the free coefficients $( \mf{h}^{(0)}, \mf{h}^{(n)} )$ in the FG expansion uniquely determine the metric $h$ on $N$}.
The key technical tool in the proof is a unique continuation result of Mazzeo \cite{mazz:uc_inf_ah} from the conformal boundary of asymptotically hyperbolic manifolds.

A similar correspondence result, also for (Riemannian) asymptotically hyperbolic Einstein manifolds, was established independently by Anderson and Herzlich; see \cite{and_herz:uc_ricci, and_herz:uc_ricci_err}.
Notice that in each of the results \cite{and_herz:uc_ricci, and_herz:uc_ricci_err, biq:uc_einstein}, the key PDEs are always elliptic in nature.

Next, the methods of \cite{biq:uc_einstein} were adapted by Chrusciel and Delay \cite{chru_delay:uc_killing} to \emph{Lorentzian} aAdS spacetimes, under the specific restriction that the spacetime is \emph{stationary}.
More specfically, \cite{chru_delay:uc_killing} showed that \emph{if $( \mi{M}, g )$ is time-independent, then $g$ is indeed uniquely determined by its holographic data $( \mi{I}, \gb{0}, \gb{n} )$}.
Furthermore, using similar techniques, \cite{chru_delay:uc_killing} positively addressed Problem \ref{prb.symmetry} in stationary settings---\emph{if $( \mi{M}, g )$ is time-independent, then any Killing vector field on the conformal boundary that also preserves $\gb{n}$ necessarily extends into $\mi{M}$ as a Killing field for $g$}.

\begin{remark}
Even though \cite{chru_delay:uc_killing} treated Lorentzian settings, the assumption of time-independence meant that all the key PDEs are elliptic.
Thus, there remains a fundamental gap between \cite{chru_delay:uc_killing} and the questions posed here in general dynamical settings, which instead involve hyperbolic PDEs.
\end{remark}

In addition, we mention, again in the Riemannian context, the classical result of Graham and Lee \cite{gra_lee:fg_riem_ball}, which proved the existence of asymptotically hyperbolic Einstein metrics on the Poincar\'e ball $\mathbb{B}^n$, with prescribed metric $\mf{h}^{(0)}$ on the conformal boundary $\Sph^{n-1}$ that is sufficiently close to the round metric.
Observe, however, that the above corresponds to solving an elliptic Dirichlet problem, which has no analogue in hyperbolic contexts.

We now turn our attention to dynamical Lorentzian settings.
First, Anderson \cite{and:uc_ads} partly showed that if a vacuum aAdS spacetime $( \mi{M}, g )$ has a stationary conformal boundary and becomes asymptotically stationary as $t \rightarrow \pm \infty$, then $( \mi{M}, g )$ must itself be stationary.
However, the result in \cite{and:uc_ads} is \emph{conditional}, in that \emph{it assumes that a unique continuation property holds for the linearized Einstein equations from the conformal boundary}.
In particular, this unique continuation condition was crucially used to prove a symmetry extension result, in the sense of Problem \ref{prb.symmetry}, for timelike Killing vector fields from the conformal boundary into the bulk spacetime.

\begin{remark}
The main symmetry extension result of this article, Theorem \ref{thm.killing}, can be applied in \cite{and:uc_ads} to prove an improved, \emph{unconditional} version of Anderson's theorem.
\end{remark}

The first concrete steps toward resolving Problem \ref{prb.correspondence}, as well as a precursor to the main result of this article, was provided by McGill \cite{mcgill:loc_ads}, which gave a holographic characterization of locally AdS spacetimes.
Very roughly, \cite{mcgill:loc_ads} established the following:

\begin{theorem}[\cite{mcgill:loc_ads}] \label{thm.correspondence_ads}
Let $( \mi{M}, g )$ be a vacuum FG-aAdS segment with the AdS conformal boundary $( \mi{I}_\textrm{AdS}, g_\textrm{AdS} )$.
Then, $( \mi{M}, g )$ is locally isometric to AdS spacetime (at least close to the conformal boundary) if and only if both $\gb{0}$ is conformally flat and $\gb{n} = 0$ on a sufficiently large timespan of the conformal boundary.
Furthermore, the above criteria for $\gb{0}$ and $\gb{n}$ are gauge-invariant, in the sense that they are preserved by gauge transformations.
\end{theorem}

McGill's result was proved using the same setting and tools \cite{hol_shao:uc_ads, hol_shao:uc_ads_ns, mcgill_shao:psc_aads, shao:aads_fg} as our main theorems, most notably the Carleman estimates for tensorial waves near aAdS conformal boundaries.
One key step is a unique continuation result for the spacetime Weyl curvature, which satisfies a tensorial wave equation, and whose vanishing ensures the spacetime is locally AdS.
In particular, \cite{mcgill:loc_ads} represented the first application of the Carleman estimates of \cite{hol_shao:uc_ads, hol_shao:uc_ads_ns, mcgill_shao:psc_aads} to the EVE, and it demonstrated that these estimates are indeed powerful enough to treat the EVE.

More generally, symmetry extension results for Killing vector fields have seen several applications in general relativity, outside of aAdS settings.
For instance, Killing extension theorems have played a crucial role in recent black hole rigidity results \cite{alex_io_kl:hawking_anal, alex_io_kl:unique_bh, alex_io_kl:rigid_bh, gior:killing_em, io_kl:killing, lai_li_yu:bh_kerr_newman, yu:hawking_em}, which showed, in various circumstances, that vacuum (or electrovacuum) stationary spacetimes must be isometric to Kerr (or Kerr-Newman) spacetime.
Next, \cite{peters:cpt_cauchy} proved a Killing extension theorem from compact Cauchy horizons, which can then be applied to questions in both black hole rigidity and cosmic censorship.
Furthermore, the results in \cite{alex_schl:time_periodic} proved non-existence of time-periodic vacuum asymptotically flat spacetimes; a key step is a Killing extension theorem from spacelike and null infinity.

\subsection{Correspondence Results}

The main result of this article provides an \emph{affirmative answer to Problem \ref{prb.correspondence}, provided the conformal boundary also satisfies a gauge-invariant geometric condition}.
We refer to this extra assumption, which in its current form was first identified and formulated by Chatzikaleas and the author in \cite{chatz_shao:uc_ads_gauge}, as the \emph{generalised null convexity criterion}, or \emph{GNCC}.

We will define and further discuss the GNCC (which can seem quite technical at first glance) in the following subsection.
But before doing so, let us first state our main result:

\begin{theorem}[Bulk-boundary correspondence \cite{hol_shao:uc_ads_eve}] \label{thm.correspondence}
Let $n > 2$, and let $( \mi{M}, g )$ and $( \bar{\mi{M}}, \bar{g} )$ denote vacuum FG-aAdS segments, with associated holographic data $( \mi{I}, \gb{0}, \gb{n} )$ and $( \bar{\mi{I}}, \gbc{0}, \gbc{n} )$, respectively.
In addition, fix $\mi{D} \subset \mi{I}$ with $\bar{\mi{D}}$ compact, and assume $( \mi{D}, \gb{0} )$ satisfies the GNCC.

Then, if $( \mi{I}, \gb{0}, \gb{n} )$ and $( \bar{\mi{I}}, \gbc{0}, \gbc{n} )$ are $\phi$-gauge-equivalent on $\mi{D}$, then there exist two neighborhoods $\mi{U} \subset \mi{M}$ and $\bar{\mi{U}} \subset \bar{\mi{M}}$---of $\{0\} \times \mi{D}$ and $\{0\} \times \phi ( \mi{D} )$, respectively---as well as an isometry $\Phi: ( \mi{U}, g ) \leftrightarrow ( \bar{\mi{U}}, \bar{g} )$ whose limit at the conformal boundary is $\phi$.
\end{theorem}

In particular, Theorem \ref{thm.correspondence} shows that the near-boundary vacuum aAdS geometries (up to isometry) are in one-to-one correspondence with the space of holographic data (up to gauge-equivalence).

\begin{remark}
That $\mi{D}$ in Theorem \ref{thm.correspondence} satisfies the GNCC can be interpreted as $\mi{D}$ being ``sufficiently large" compared to the geometry of $( \mi{I}, \gb{0} )$.
In other words, the correspondence result holds only when the holographic data match (up to gauge) on a sufficiently large boundary region.
\end{remark}

\begin{remark}
As GNCC is a gauge-invariant condition (see Proposition \ref{thm.gncc_gauge} below), then $( \mi{D}, \gb{0} )$ satisfying the GNCC is equivalent to $( \phi ( \mi{D} ), \gbc{0} )$ satisfying the GNCC.
In other words, all assumptions on the conformal boundaries in Theorem \ref{thm.correspondence} are gauge-invariant.
\end{remark}

\begin{remark}
Note the correspondence described in Theorem \ref{thm.correspondence} is local, as the FG-aAdS segments are only isometric near the subregions $\mi{D}$ and $\phi ( \mi{D} )$ of the conformal boundaries.
In particular, one needs additional information to prove a global correspondence.
\end{remark}

\begin{remark}
By pulling back through $\phi$, we can, without loss of generality, reduce Theorem \ref{thm.correspondence} to the case in which $\bar{\mi{M}} = \mi{M}$, $\bar{\mi{I}} = \mi{I}$, and $\phi$ is the identity map on $\mi{I}$; see the beginning of Section \ref{sec.proof}.
In fact, the main result of \cite{hol_shao:uc_ads_eve} is only formulated in this special case.
We state the general case in Theorem \ref{thm.correspondence} in order to highlight the geometric and gauge-covariant nature of the result.
\end{remark}

\begin{remark}
The main result Theorem \ref{thm.correspondence_ads} of \cite{mcgill:loc_ads} can be viewed as a special case of Theorem \ref{thm.correspondence}, for spacetimes that are locally isometric to AdS.
Furthermore, the assumption in Theorem \ref{thm.correspondence_ads} of ``sufficiently large timespan of the conformal boundary" is a direct consequence of the GNCC.
\end{remark}

Next, one important application of Theorem \ref{thm.correspondence} is toward proving symmetry extension results on aAdS spacetimes, which address Problem \ref{prb.symmetry}.
The first such result deals with extending discrete symmetries from the conformal boundary into the bulk spacetime:

\begin{definition} \label{def.symmetry}
Let $( \mi{I}, \gb{0}, \gb{n} )$ denote holographic data, and fix $\mi{D} \subset \mi{I}$ with $\bar{\mi{D}}$ compact.
We say that $\psi: \mi{I}' \rightarrow \mi{I}$, where $\bar{\mi{D}} \subset \mi{I}' \subset \mi{I}$, is a \emph{holographic gauge symmetry} on $\mi{D}$ iff $( \mi{I}', \gb{0}, \gb{n} )$ and $( \psi( \mi{I}' ), \gb{0}, \gb{n} )$ are $\psi$-gauge-equivalent on $\mi{D}$.
If $\psi$ satisfies, in addition,
\begin{equation}
\label{eq.symmetry} \psi_\ast \gb{0} |_{ \mi{D} } = \gb{0} |_{ \mi{D} } \text{,} \qquad \psi_\ast \gb{n} |_{ \mi{D} } = \gb{n} |_{ \mi{D} } \text{.}
\end{equation}
then $\psi$ is called a \emph{holographic symmetry} on $\mi{D}$.
\end{definition}

Informally, a holographic gauge symmetry is a diffeomorphism within $\mi{I}$ that preserves $\gb{0}$ and $\gb{n}$ up to gauge-equivalence, while a holographic symmetry preserves $\gb{0}$ and $\gb{n}$ entirely.

\begin{theorem}[Extension of symmetries \cite{hol_shao:uc_ads_eve}] \label{thm.symmetry}
Let $n > 2$, and let $( \mi{M}, g )$ be a vacuum FG-aAdS segment, with holographic data $( \mi{I}, \gb{0}, \gb{n} )$.
Also, fix $\mi{D} \subset \mi{I}$ with $\bar{\mi{D}}$ compact, and assume $( \mi{D}, \gb{0} )$ satisfies the GNCC.
Then, if $\psi$ is a holographic gauge symmetry on $\mi{D}$, then there is a neighborhood $\mi{U} \subset \mi{M}$ of $\{0\} \times \mi{D}$ such that $\psi$ extends to an isometry $\Psi$ between $( \mi{U}, g )$ and $( \Psi ( \mi{U} ), g )$.
\end{theorem}

\begin{proof}
This follows immediately from Theorem \ref{thm.correspondence}---assuming $\psi: \mi{I}' \rightarrow \mi{I}$, where $\bar{\mi{D}} \subset \mi{I}' \subset \mi{I}$, then we simply apply Theorem \ref{thm.correspondence} to the vertical truncations of $\mi{M}$ to $\mi{I}'$ and $\psi ( \mi{I}' )$.
\end{proof}

\begin{remark}
If $\psi$ in Theorem \ref{thm.symmetry} is, in addition, a holographic symmetry on $\mi{D}$, then the desired isometry $\Psi$ can be given explicitly in terms of $\psi$ by
\begin{equation}
\label{eq.symmetry_extension} \Psi ( \rho, p ) = ( \rho, \psi (p) ) \text{,} \qquad ( \rho, p ) \in \mi{U} \subset \mi{M} \text{.}
\end{equation}
\end{remark}

Another variant (as well as a direct corollary) of Theorem \ref{thm.symmetry} concerns extending a smooth family of symmetries, in the form of a Killing vector field.
This is stated in the subsequent theorem, which, on its own, can also be seen as an analogue of the symmetry extension results \cite{alex_io_kl:hawking_anal, alex_schl:time_periodic, gior:killing_em, io_kl:killing, peters:cpt_cauchy}, but in the setting of conformal boundaries of vacuum aAdS spacetimes.

\begin{definition} \label{def.killing}
Let $( \mi{I}, \gb{0}, \gb{n} )$ denote holographic data, and fix $\mi{D} \subset \mi{I}$ with $\bar{\mi{D}}$ compact.
We say that a vector field $\mf{K}$ on $\mi{I}'$, where $\bar{\mi{D}} \subset \mi{I}' \subset \mi{I}$, is \emph{holographic gauge Killing} on $\mi{D}$ iff there exist holographic data $( \bar{\mi{I}}, \gbc{0}, \gbc{n} )$ that is $\phi$-gauge-equivalent to $( \mi{I}, \gb{0}, \gb{n} )$ and such that
\begin{equation}
\label{eq.conformal_killing} \mi{L}_{ \phi^\ast \mf{K} } \gbc{0} |_{ \phi ( \mi{D} ) } = 0 \text{,} \qquad \mi{L}_{ \phi^\ast \mf{K} } \gbc{n} |_{ \phi ( \mi{D} ) } = 0 \text{.}
\end{equation}
Furthermore, $\mf{K}$ is called \emph{holographic Killing} on $\mi{D}$ iff
\begin{equation}
\label{eq.killing} \mi{L}_{ \mf{K} } \gb{0} |_{ \mi{D} } = 0 \text{,} \qquad \mi{L}_{ \mf{K} } \gb{n} |_{ \mi{D} } = 0 \text{.}
\end{equation}
\end{definition}

\begin{remark}
As is standard, $\phi^\ast \mf{K}$ denotes the push-forward of $\mf{K}$ through $\phi$.
\end{remark}

Informally, a holographic Killing vector field $\mf{K}$ preserves $\gb{0}$ and $\gb{n}$ along its integral curves, while a holographic gauge Killing field only preserves $\gb{0}$ and $\gb{n}$ up to gauge-equivalence.

\begin{theorem}[Extension of Killing vector fields \cite{hol_shao:uc_ads_eve}] \label{thm.killing}
Let $n > 2$, and let $( \mi{M}, g )$ be a vacuum FG-aAdS segment, with holographic data $( \mi{I}, \gb{0}, \gb{n} )$.
Fix $\mi{D} \subset \mi{I}$ with $\bar{\mi{D}}$ compact, and assume $( \mi{D}, \gb{0} )$ satisfies the GNCC.
Then, if $\mf{K}$ is a holographic gauge Killing vector field on $\mi{D}$, then there is a neighborhood $\mi{U} \subset \mi{M}$ of $\{0\} \times \mi{D}$ such that $\mf{K}$ extends to a ($g$-)Killing vector field $K$ on $\mi{U}$.
\end{theorem}

\begin{proof}
By pushing through the gauge transformation $\phi$ in Definition \ref{def.killing} and relying on Theorem \ref{thm.correspondence}, we can reduce to the case in which $\mf{K}$ is holographic Killing on $\mi{D}$ and \eqref{eq.killing} holds. 
Let $\psi_s$, for small $s \in \R$, be the family of holographic symmetries generated by transporting along the integral curves of $\mf{K}$.
By Theorem \ref{thm.symmetry}, these extend to a family of $g$-isometries $\Psi_s$ in the bulk spacetime.
The Killing vector field $K$ is then obtained as the generator of the isometries $\Psi_s$.
\end{proof}

One immediate consequence of Theorem \ref{thm.killing} and the classical Birkhoff theorem is the following holographic rigidity result for Schwarzschild-AdS spacetimes:

\begin{corollary}[Rigidity of Schwarzschild-AdS spacetime] \label{thm.sads}
Let $n > 2$, and let $( \mi{M}, g )$ be a vacuum FG-aAdS segment, with holographic data $( \mi{I}_\textrm{AdS}, \gb{0}, \gb{n} )$.
Also, fix $\mi{D} \subset \mi{I}$ with $\bar{\mi{D}}$ compact, and assume $( \mi{D}, \gb{0} )$ satisfies the GNCC.
If $( \gb{0}, \gb{n} )$ is spherically symmetric on $\mi{D}$, then there is a neighborhood $\mi{U}$ of $\{0\} \times \mi{D}$ such that $( \mi{U}, g )$ is isometric to part of a Schwarzschild-AdS spacetime.
\end{corollary}

\begin{remark}
If one takes the AdS conformal boundary, $\gb{0} = \gm_\textrm{AdS}$, in Corollary \ref{thm.sads}, then its conclusions hold with $\mi{D} := \{ t_- < t < t_+ \}$ whenever $t_+ - t_- > \pi$; see Proposition \ref{thm.gncc_ads} below.
\end{remark}

Finally, note that Theorem \ref{thm.symmetry} also can be applied toward discrete symmetries that are not generated by Killing vector fields.
In particular, one consequence of Theorem \ref{thm.symmetry} is that time periodicity of the conformal boundary is inherited by the bulk spacetime:

\begin{corollary}[Extension of time periodicity] \label{theo:symmetries_periodic}
Let $n > 2$, and let $( \mi{M}, g )$ denote a vacuum FG-aAdS segment, with holographic data $( \mi{I}, \gb{0}, \gb{n} )$.
In addition, fix $\mi{D} \subset \mi{I}$ with $\bar{\mi{D}}$ compact, and assume that $( \mi{D}, \gb{0} )$ satisfies the GNCC.
If $( \gb{0}, \gb{n} )$ is time-periodic on $\mi{D}$, then there exists a neighborhood $\mi{U}$ of $\{0\} \times \mi{D}$ such that $( \mi{U}, g )$ is also time-periodic.
\end{corollary}

\subsection{Null Convexity Criteria}

Having stated our key results, we now turn our attention toward the crucial geometric assumption required for these results---the GNCC of \cite{chatz_shao:uc_ads_gauge}.
The first task is to state the GNCC, adapted here to the special case of vacuum aAdS spacetimes:

\begin{definition}[GNCC \cite{chatz_shao:uc_ads_gauge}] \label{def.gncc}
Let $( \mi{I}, \gm )$ be an $n$-dimensional Lorentzian manifold, and fix any Riemannian metric $\mf{h}$ on $\mi{I}$.
In addition, fix an open subset $\mi{D} \subset \mi{I}$, with $\bar{\mi{D}}$ compact.
We say that $( \mi{D}, \mf{g} )$ satisfies the \emph{generalized null convexity criterion} (abbreviated \emph{GNCC}) iff there is a $C^4$-function $\eta$, defined on a neighborhood of $\bar{\mi{D}}$, such that the following hold:
\begin{itemize}
\item $\eta > 0$ on $\mi{D}$.

\item $\eta = 0$ on the boundary $\partial \mi{D}$.

\item There exists $c > 0$ such that the following holds on $\mi{D}$ for any $\mf{g}$-null vector field $\mf{Z}$,
\begin{equation}
\label{eq.gncc} ( \Dm^2 \eta + \eta \cdot \mc{P} [ \gm ] ) ( \mf{Z}, \mf{Z} ) > c \eta \cdot \mf{h} ( \mf{Z}, \mf{Z} ) \text{,}
\end{equation}
where $\Dm^2$ and $\mc{P} [ \gm ]$ denote the Hessian and Schouten tensor with respect to $\gm$, respectively.
\end{itemize}
\end{definition}

\begin{remark}
Note Definition \ref{def.gncc} is independent of the choice of $\mf{h}$, since $\bar{\mi{D}}$ is compact.
Moreover, observe that $\mc{P} [ \gm ]$ in \eqref{eq.gncc} can be replaced by $\smash{ \frac{1}{n-2} \operatorname{Ric} [ \gm ] }$, since their difference is proportional to $\gm$---see \eqref{eq.fg_schouten}---and hence vanishes along all $\mf{g}$-null directions.
\end{remark}

\begin{remark}
The GNCC can be extended to FG-aAdS segments that are not vacuum by replacing the Schouten tensor $\mc{P} [ \gm ]$ in \eqref{eq.gncc} with $- \gb{2}$.
(Note the general case reduces to Definition \ref{def.gncc} in vacuum settings due to \eqref{eq.fg_schouten}.)
This is the GNCC that was stated in \cite{chatz_shao:uc_ads_gauge}.
\end{remark}

One important feature of the GNCC is that it is a conformally invariant property.
In particular, \cite[Proposition 3.6]{chatz_shao:uc_ads_gauge} showed, in the context of Definition \ref{def.gncc}, that \emph{if $( \mi{D}, \gm )$ satisfies the GNCC, with associated function $\eta$, then $( \mi{D}, e^{ 2 \mf{a} } \gm )$ also satisfies the GNCC, with associated function $\eta' := e^\mf{a} \eta$}.
We can immediately rephrase the above in terms of our holographic gauge-invariance.

\begin{proposition} \label{thm.gncc_gauge}
Let $( \mi{I}, \gb{0}, \gb{n} )$, $( \bar{\mi{I}}, \gbc{0}, \gbc{n} )$ denote holographic boundary data, and suppose $( \mi{I}, \gb{0}, \gb{n} )$, $( \bar{\mi{I}}, \gbc{0}, \gbc{n} )$ are $\phi$-gauge-equivalent.
Moreover, let $\mi{D} \subset \mi{I}$ be open, with $\bar{\mi{D}}$ compact.
Then, $( \mi{D}, \gb{0} )$ satisfies the GNCC if and only if $( \phi ( \mi{D} ), \gbc{0} )$ satisfies the GNCC.
\end{proposition}

The main significance of the GNCC, demonstrated in \cite{chatz_shao:uc_ads_gauge}, is that it provides the crucial conditions on the conformal boundary that lead to pseudoconvexity in the near-boundary aAdS geometry.
This observation was a crucial ingredient in proving the Carleman estimates of \cite{chatz_shao:uc_ads_gauge} for tensorial waves on FG-aAdS segments.
These Carleman estimates, in turn, serve as a key tool in the proof of Theorem \ref{thm.correspondence}.
See Sections \ref{sec.carleman} and \ref{sec.proof} for further discussions of these points.

Definition \ref{def.gncc} may seem, at first glance, rather abstract and technical.
Furthermore, the conditions for the GNCC, which involve finding a desirable candidate function $\eta$, can be rather difficult to check in practice.
Therefore, to further flesh out Definition \ref{def.gncc}, we consider, in the remainder of this subsection, the GNCC applied to some special cases of interest.

Let us begin with the standard AdS conformal boundary $( \mi{I}_\textrm{AdS}, \gm_\textrm{AdS} )$ of \eqref{eq.ads_boundary}, which satisfies
\begin{equation}
\label{eq.ads_P} \gm_\textrm{AdS} = - dt^2 + \mathring{\gamma} \text{,} \qquad \mc{P} [ \gm_\textrm{AdS} ] = \frac{1}{2} ( dt^2 + \mathring{\gamma} ) \text{.}
\end{equation}
In addition, we take $\mi{D} := \mi{D}_0$ to be the time slab
\begin{equation}
\label{eq.ads_D} \mi{D}_0 := ( t_-, t_+ ) \times \Sph^{n-1} \text{,} \qquad \partial \mi{D}_0 = \{ t_-, t_+ \} \times \Sph^{n-1} \text{,} \qquad t_- < t_+ \text{.}
\end{equation}

The following observation was established in \cite[Corollary 3.14]{chatz_shao:uc_ads_gauge}:

\begin{proposition}[\cite{chatz_shao:uc_ads_gauge, hol_shao:uc_ads}] \label{thm.gncc_ads}
$( \mi{D}_0, \gm_\textrm{AdS} )$ satisfies the GNCC if and only if $t_+ - t_- > \pi$.
\end{proposition}

\begin{proof}
The key observation is that by taking $\eta$ to depend only on $t$, then \eqref{eq.gncc} and \eqref{eq.ads_P} yield
\begin{equation}
\label{eq.ads_eta} ( \Dm^2 \eta + \eta \cdot \mc{P} [ \gm ] ) ( \mf{Z}, \mf{Z} ) = ( \mf{Z} t )^2 \, ( \ddot{\eta} + \eta ) \text{.}
\end{equation}
Then, one can directly check that the choice
\[
\eta (t) := \sin \left( \pi \cdot \frac{ t - t_- }{ t_+ - t_- } \right)
\]
satisfies the conditions of Definition \ref{def.gncc} whenever $t_+ - t_- > \pi$.
(Observe $t_+ - t_- > \pi$ is needed for the right-hand side of \eqref{eq.ads_eta} to be positive when $t_- < t < t_+$.)

For the converse, the basic idea is that one can apply a standard Sturm comparison argument with the function $\sin ( \pi t )$ to show that if $\eta := \eta (t)$ satisfies
\[
\eta ( t_\pm ) = 0 \text{,} \qquad \eta |_{ ( t_-, t_+ ) } > 0 \text{,} \qquad \ddot{\eta} + \eta > 0 \text{,}
\]
then $t_+ - t_- > \pi$.
In particular, this, along with \eqref{eq.ads_eta}, shows that no function of the form $\eta = \eta (t)$ can satisfy the conditions of Definition \ref{def.gncc} whenever $t_+ - t_- \leq \pi$.
Moreover, the above argument can be adapted to general functions $\eta$, from which one concludes that $( \mi{D}_0, \gm_\textrm{AdS} )$ cannot satisfy the GNCC whenever $t_+ - t_- \leq \pi$; see \cite[Lemma 3.7]{chatz_shao:uc_ads_gauge} for details.
\end{proof}

\begin{remark}
Note that Proposition \ref{thm.gncc_ads} applies to AdS spacetime, as well as every Schwarzschild-AdS and Kerr-AdS spacetime, since these all induce the AdS conformal boundary.
\end{remark}

\begin{remark}
The key consequence of Proposition \ref{thm.gncc_ads}---that unique continuation for wave equations holds from $\mi{D}_0$ when $t_+ - t_- > \pi$---was first proved as a special case of the results of \cite{hol_shao:uc_ads}.
\end{remark}

Next, we consider general conformal boundary domains that are foliated by a time function $t$,
\begin{equation}
\label{eq.intro_time} \mi{I}_\ast := \R_t \times \mc{S} \text{,} \qquad \mi{D}_\ast := ( t_-, t_+ ) \times \mc{S} \text{,}
\end{equation}
with $\mc{S}$ a compact manifold of dimension $n - 1$.
Previous versions of Carleman estimates and unique continuation results for wave equations were developed in the setting \eqref{eq.intro_time}, under conditions on the conformal boundary that can be viewed as special cases of the GNCC.

First, Holzegel and the author \cite{hol_shao:uc_ads} treated the case in which the metric $\gm$ is static (with time function $t$), establishing Carleman estimates and unique continuation under a so-called \emph{pseudoconvexity criterion}.
In particular, this was the first unique continuation result for wave equations from aAdS conformal boundaries.
The result was subsequently extended by Holzegel and the author \cite{hol_shao:uc_ads_ns} to some non-static $\gm$, and then by McGill and the author \cite{mcgill_shao:psc_aads} to a wider class of non-static $\gm$ and time foliations and under a slightly weaker conformal boundary criterion.

Let us focus on the key criterion of \cite{mcgill_shao:psc_aads}, as well as its relation to the GNCC:

\begin{proposition}[\cite{chatz_shao:uc_ads_gauge, mcgill_shao:psc_aads}] \label{thm.ncc}
Assume the setting of \eqref{eq.intro_time}, and let $\gm$ be a Lorentzian metric on $\mi{I}_\ast$.
Suppose there exist $0 \leq B < C$ such that the following holds for any $\mf{g}$-null vector field $\mf{Z}$:
\begin{equation}
\label{eq.ncc} \mc{P} [ \gm ] ( \mf{Z}, \mf{Z} ) \geq C^2 \cdot ( \mf{Z} t )^2 \text{,} \qquad | \mf{D}^2 t ( \mf{Z}, \mf{Z} ) | \leq 2 B \cdot ( \mf{Z} t )^2 \text{.}
\end{equation}
Then, \emph{$( \mi{D}_\ast, \gm )$ satisfies the GNCC} as long as $t_+ - t_-$ is large enough (depending on $B$ and $C$).
\end{proposition}

\begin{proof}
The process is similar to the proof of Proposition \ref{thm.gncc_ads}, except one now chooses $\eta$ (still depending only on $t$) to roughly solve a damped harmonic oscillator:
\begin{equation}
\label{odel} \ddot{\eta} - 2 b | \dot{\eta} | + c^2 \eta = 0 \text{,} \qquad B \leq b < c < C \text{.}
\end{equation}
The reader is referred to \cite[Proposition 3.13]{chatz_shao:uc_ads_gauge} for details.
\end{proof}

\begin{remark}
The second condition of \eqref{eq.ncc} can be viewed as a bound on the non-stationarity of $\gm$, since $\mf{D}^2 t$ is proportional to the Lie derivative of $\gm$ along the gradient of $t$.
\end{remark}

\begin{remark}
The connection between damped harmonic oscillators \eqref{odel} and unique continuation from $\mi{D}_\ast$ was first illuminated in \cite{hol_shao:uc_ads_ns}; see the discussions therein.
\end{remark}

\begin{remark}
Proposition \ref{thm.ncc} can be used to generalize the conclusions of Proposition \ref{thm.gncc_ads}:
\begin{itemize}
\item For instance, if $t_+ - t_- > \pi$, then Proposition \ref{thm.ncc} implies that $( \mi{D}_0, \gm )$ satisfies the GNCC whenever $\gm$ is a sufficiently small perturbation of $\gm_\textrm{AdS}$.

\item If $\gm$ is static with respect to $t$, and if the cross-sections $\mc{S}$ have positive Ricci curvature, then $( \mi{D}_\ast, \gm )$ satisfies the GNCC for sufficiently large $t_+ - t_-$; see \cite[Proposition B.2]{hol_shao:uc_ads_ns}.
\end{itemize}
\end{remark}

The conditions \eqref{eq.ncc} were first identified in \cite{mcgill_shao:psc_aads} and were called the \emph{null convexity criterion} (or \emph{NCC}).
In particular, \cite{mcgill_shao:psc_aads} established Carleman estimates and unique continuation results, under the assumption that the NCC is satisfied.
Note that Proposition \ref{thm.ncc} shows that the GNCC indeed generalizes the NCC.
Furthermore, the GNCC removes the need for a predetermined time function and allows for a larger class of boundary domains $\mi{D}$ to be treated.

One advantage of the NCC \eqref{eq.ncc} is that it is more concrete and easier to check than the GNCC.
On the other hand, one shortcoming of the NCC is that \eqref{eq.ncc} fails to be gauge-invariant, since a conformal transformation of $\gm$ can cause \eqref{eq.ncc} to no longer hold.
This makes the NCC undesirable for Theorem \ref{thm.correspondence}, and this served as a key motivation for developing the GNCC.

\subsection{Geodesic Return}

While the GNCC provides a sufficient condition for Theorem \ref{thm.correspondence}, due to its connection to pseudoconvexity and to Carleman estimates for wave equations, one can also make a heuristic case that the GNCC may be a necessary condition for Theorem \ref{thm.correspondence}.

The necessity of some geometric condition in Theorem \ref{thm.correspondence} was first conjectured in \cite{hol_shao:uc_ads, hol_shao:uc_ads_ns}, due to the special properties of AdS geometry near its conformal boundary $\mi{I}_\textrm{AdS}$.
More specifically, on AdS spacetime, \emph{there exist null geodesics which initiate from $\mi{I}_\textrm{AdS}$ at time $t = 0$, remain arbitrarily close to $\mi{I}_\textrm{AdS}$, and then terminate at the boundary at time $t = \pi$}; see \cite[Section 1.2]{hol_shao:uc_ads}.

\begin{remark}
In terms of the standard conformal embedding of AdS spacetime into the Einstein cylinder $\R_t \times \Sph^n_+$, these null geodesics can be explicitly described as
\[
\Lambda: ( 0, \pi ) \rightarrow \R \times \Sph^n_+ \text{,} \qquad \Lambda ( \tau ) = ( \tau, \lambda (\tau) ) \text{,}
\]
where $\lambda$ denotes unit speed parametrizations of great circles in $\Sph^n_+$, with the limits $\lambda (0+)$ and $\lambda (\pi-)$ lying on antipodal points of the equator $\partial \Sph^n_+$.
The near-boundary trapping effect caused by these $\Lambda$ is also connected to the relative lack of decay for waves on AdS spacetime; see, e.g., \cite{hol_smul:decay_kg_ads}.
\end{remark}

From the above, one can construct, via the geometric optics methods of Alinhac and Baouendi \cite{alin_baou:non_unique}, local solutions to linear wave equations on AdS spacetime that propagate along these null geodesics.
Most notably, these solutions vanish to infinite order toward the portion $( 0, \pi ) \times \Sph^{n-1}$ of the conformal boundary, but immediately become nonzero away from this boundary segment.
This yields, for AdS spacetime, counterexamples to unique continuation for various wave equations when the data on the conformal boundary is imposed on a timespan of less than $\pi$---the return time of the above-mentioned near-boundary null geodesics.

On the other hand, if we consider a slab $\mi{D}_0$ as in \eqref{eq.ads_D}, with timespan $t_+ - t_- > \pi$, then none of the aforementioned null geodesics (translated forward or backward in time as needed) can travel entirely over $\mi{D}_0$.
In fact, if such a geodesic goes over $\mi{D}_0$ at all, then it must either start from or terminate at $\mi{D}_0$.
From this, one concludes that the Alinhac-Baouendi counterexamples cannot be constructed over the time slab $\mi{D}_0$ whenever $t_+ - t_- > \pi$.

\begin{remark}
There are multiple important caveats regarding the preceding counterexamples.
The first is that the methods of \cite{alin_baou:non_unique} only apply directly to wave operators $\Box_g + \sigma$ with $4 \sigma := n^2 - 1$.
(For other values of $\sigma$, the corresponding wave equation in the conformally compactified setting gains a potential that becomes critically singular at the conformal boundary, which significantly complicates the constructions.)
The case of general $\sigma$ will be treated in the upcoming work of Guisset \cite{guiss:non_unique}.
\end{remark}

\begin{remark}
Another caveat is that not every linear wave equation on AdS spacetime can have these counterexamples to unique continuation.
For example, from Holmgren's theorem \cite{holmg:uc_anal, hor:lpdo2}, if all the coefficients of the wave equation are real-analytic (in particular, this applies to $( \Box_g + \sigma ) u = 0$), then counterexamples to unique continuation cannot exist.
In particular, the constructions of \cite{alin_baou:non_unique} only show that there exist potentials $V$ (vanishing at the conformal boundary) such that
\[
( \Box_g + \sigma + V ) u = 0
\]
has the above-mentioned counterexamples to unique continuation, but it does not address whether such counterexamples exist for any fixed, chosen potential $V$.
\end{remark}

One can in fact view the GNCC as a generalization of the above intuitions for AdS spacetime to aAdS settings.
This was observed by Chatzikaleas and the author in \cite[Theorem 4.1]{chatz_shao:uc_ads_gauge}, which connected the GNCC to the trajectories of spacetime null geodesics near the conformal boundary.
We restate the result here, but slightly adapted to the special case of vacuum FG-aAdS segments:

\begin{theorem}[Geodesic return \cite{chatz_shao:uc_ads_gauge, mcgill_shao:psc_aads}] \label{thm.geodesic}
Let $( \mi{M}, g )$ be a vacuum FG-aAdS segment, with holographic data $( \mi{I}_\textrm{AdS}, \gb{0}, \gb{n} )$, and fix $\mi{D} \subset \mi{I}$, with $\bar{\mi{D}}$ compact, so that $( \mi{D}, \gb{0} )$ satisfies the GNCC.
In addition, let $\Lambda: ( s_-, s_+ ) \longrightarrow \mi{M}$ be a complete null geodesic with respect to $\rho^2 g$, written as
\[
\Lambda (s) := ( \rho (s), \lambda (s) ) \in ( 0, \rho_0 ) \times \mi{I} \text{,} \qquad s \in ( s_-, s_+ ) \text{,}
\]
with $s$ being an affine parameter for $\Lambda$.
Assume also that 
\begin{equation}
\label{eq.geodesic_ass} 0 < \rho (s_0) < \epsilon_0 \text{,} \qquad | \dot{\rho} (s_0) | \lesssim \rho (s_0) \text{,} \qquad \lambda (s_0) \in \mi{D} \text{.}
\end{equation}
for some $s_0 \in ( s_-, s_+ )$, where $\epsilon_0 > 0$ is sufficiently small (depending on $\ms{g}$ and $\mi{D}$).
Then,
\begin{itemize}
\item either $\Lambda$ initiates from the conformal boundary within $\mi{D}$,
\begin{equation}
\label{eq.geodesic_start} \lim_{ s \searrow s_- } \rho (s) = 0 \text{,} \qquad \lim_{ s \searrow s_- } \lambda (s) \in \mi{D} \text{,}
\end{equation}

\item or $\Lambda$ terminates at the conformal boundary within $\mi{D}$,
\begin{equation}
\label{eq.geodesic_end} \lim_{ s \nearrow s_+ } \rho (s) = 0 \text{,} \qquad \lim_{ s \nearrow s_+ } \lambda (s) \in \mi{D} \text{.}
\end{equation}
\end{itemize}
\end{theorem}

\begin{remark}
Since $g$ and $\rho^2 g$ have the same null geodesics, we can more usefully interpret $\Lambda$ in Theorem \ref{thm.geodesic} as being a $g$-null geodesic.
However, it is more convenient to parametrize $\Lambda$ with respect to $\rho^2 g$, since the conformal boundary is finite from the perspective of $\rho^2 g$.
\end{remark}

Note that the first and third conditions of \eqref{eq.geodesic_ass} mean that $\Lambda$ is both ``hovering over $\mi{D}$" and ``$\epsilon_0$-close to $\mi{D}$" at parameter $s_0$.
The second part of \eqref{eq.geodesic_ass} can be viewed as a necessary condition in order for $\Lambda$ to remain similarly close to the conformal boundary.
Consequently, one can interpret Theorem \ref{thm.geodesic} as saying that any spacetime null geodesic $\Lambda$ that is sufficiently close to the conformal boundary and that travels over $\mi{D}$ (which satisfes the GNCC) must return to the conformal boundary within $\mi{D}$, in either the future or past direction.
In other words, there cannot exist any null geodesics sufficiently near the boundary that travel over $\mi{D}$ but do not terminate within $\mi{D}$ itself.

In summary, we conclude that \emph{analogues of the Alinhac-Baouendi counterexamples in AdS spacetime cannot be constructed over any $\mi{D}$ that satisfies the GNCC}.

\begin{remark}
The first version of the geodesic return theorem was proven by McGill and the author in \cite[Theorem 4.5]{mcgill_shao:psc_aads}.
This showed that the NCC \eqref{eq.ncc} being satisfied implies that counterexamples to unique continuation cannot be constructed over the slab $\mi{D}_\ast$ in Proposition \ref{thm.ncc}.
One can hence view Theorem \ref{thm.geodesic} as an extension of \cite[Theorem 4.5]{mcgill_shao:psc_aads} to the GNCC.
\end{remark}

To conclude, the preceding discussions provide two justifications for the GNCC being the crucial condition for unique continuation and for Theorem \ref{thm.correspondence}:
\begin{itemize}
\item The GNCC rules out the known counterexamples to unique continuation for waves.

\item The GNCC implies pseudoconvexity, leading to positive unique continuation results.
\end{itemize}

\begin{remark}
Though the GNCC is crucial to our proof of Theorem \ref{thm.correspondence}, it is not known whether the methods of \cite{alin_baou:non_unique} extend to the nonlinear EVE.
The construction of counterexamples to unique continuation in Theorem \ref{thm.correspondence} when the GNCC is violated is a challenging open problem.
\end{remark}

Finally, we note that this connection with null geodesics can, in some cases, be exploited to show that some domains fail to satisfy the GNCC.
The most notable examples involve flat conformal boundaries, in particular those of the planar and toric (Kerr-)AdS spacetimes:

\begin{proposition} \label{thm.flat_boundary}
Consider the planar AdS and toric AdS conformal boundaries \eqref{eq.flat_boundary}, both with the flat Lorentzian metric.
Then, no subdomain $\mi{D}$ of either boundary can satisfy the GNCC.
\end{proposition}

The proof is similar to that of the ``negative" half of Proposition \ref{thm.gncc_ads}; see \cite[Corollary 3.10]{chatz_shao:uc_ads_gauge}.
The key intuition is that on both planar and toric AdS spacetimes, there are spacetime null geodesics that remain arbitrarily close to but never intersect the conformal boundary for all times.

\section{Carleman Estimates} \label{sec.carleman}

As mentioned before, the main technical tool behind the proof of Theorem \ref{thm.correspondence} is a novel Carleman estimate for (vertical) tensorial wave equations near conformal boundaries of FG-aAdS segments.
Various versions of this Carleman estimate were proved, in increasing levels of generality, in \cite{hol_shao:uc_ads}, \cite{hol_shao:uc_ads_ns}, \cite{mcgill_shao:psc_aads}, and finally culminating in \cite{chatz_shao:uc_ads_gauge} (which introduced the GNCC).
In this section, we discuss the Carleman estimate \cite{chatz_shao:uc_ads_gauge}, along with some of the main ideas behind its proof.

\subsection{Unique Continuation}

Carleman estimates have long played a fundamental role in the theory of unique continuation for PDEs.
They have been spectacularly successful as tools for proving robust classes of unique continuation results that, in particular, do not require any assumption of analyticity---in either the solution or the PDE.
(This is in direct contrast to the earliest techniques, through the Cauchy-Kovalevskaya and Holmgren theorems.)

Applications of Carleman estimates to unique continuation began with the work of Carleman \cite{carl:uc_strong}, which proved strong unique continuation for elliptic equations in $2$ dimensions.
Further breakthroughs extending to elliptic PDEs in all dimensions are attributed to Calder\'on \cite{cald:unique_cauchy} (for unique continuation from an open set or a hypersurface) and Aronszajn \cite{aron:uc_strong} (for strong unique continuation, i.e., from a point).
The classical unique continuation theory for more general equations---including wave equations---is largely attributed to H\"ormander; see \cite{hor:lpdo4}.

In this classical theory, the key condition required for unique continuation is \emph{pseudoconvexity}.
Roughly speaking, the result is that \emph{if $u$ solves linear PDE $\mc{L} u = 0$, and if $u$ has zero Cauchy data on a hypersurface $\Sigma$ that is pseudoconvex with respect to $\mc{L}$ and a given side of $\Sigma$, then $u$ identically vanishes near $\Sigma$ on the side indicated by the pseudoconvexity}.
This result is proved by deriving a local Carleman estimate near $\Sigma$, from which one then obtains unique continuation via a standard argument.
Furthermore, pseudconvexity has been shown to be crucial, as counterexamples to unique continuation in its absence were constructed by Alinhac \cite{alin:non_unique} and Alinhac-Baouendi \cite{alin_baou:non_unique}.

The reader is referred to \cite{hor:lpdo4, ler:carleman_ineq} for the general definition of pseudoconvexity (in the language of microlocal analysis).
Here, we focus exclusively on wave operators on Lorentzian manifolds.
In this case, one can give a geometric characterization of pseudoconvexity:

\begin{definition}[\cite{ler_robb:unique}] \label{def.pseudoconvex}
Let $( M, g )$ be a Lorentzian manifold, and consider a hypersurface
\[
\Sigma := \{ f = 0 \} \subset M \text{,}
\]
which is a level set of a function $f$ on $M$.
We say that $\Sigma$ is \emph{pseudoconvex} (with respect to the wave operator $\Box_g$ and the direction of increasing $f$) if for any vector field $Z$ on $\Sigma$ satisfying
\begin{equation}
\label{eq.pseudoconvex_ass} Z f = 0 \text{,} \qquad g ( Z, Z ) = 0 \text{,}
\end{equation}
the following inequality holds (with $\nabla^2$ denoting the Hessian with respect to $g$):
\begin{equation}
\label{eq.pseudoconvex} \nabla^2 f ( Z, Z ) < 0 \text{.}
\end{equation}
\end{definition}

To summarize, $\Sigma$ is pseudoconvex if $-f$ is convex on $\Sigma$ with respect to all null directions that are tangent to $\Sigma$.
One qualitative interpretation of this is as follows:\ \emph{if $\Lambda$ is a null geodesic that goes through $p \in \Sigma$ and is tangent to $\Sigma$ at $p$, then $\Lambda$ lies on the side $\{ f < 0 \}$ near $p$}.

With regards to the wave operator $\Box_g$ on $M$, a Carleman estimate roughly takes the form
\begin{equation}
\label{eq.carleman_gen} \lambda \int_\Omega e^{ - \lambda F } ( | u |^2 + | \nabla u |^2 ) \, dg \lesssim \int_\Omega e^{ - \lambda F } | \Box_g u |^2 \, dg \text{.}
\end{equation}
Here, $\Omega \subset M$ is a spacetime domain, and both integrals are with respect to the volume form induced by $g$.
Thus, one can view \eqref{eq.carleman_gen} as a spacetime weighted $H^1$-bound for solutions $u$ of wave equations.
A key feature of \eqref{eq.carleman_gen} is the parameter $\lambda > 0$, which can be freely chosen as long as it is sufficiently large.
Finally, regarding the weight $e^{ - \lambda F }$, its main component is a function $F$ that is constructed from the function $f$ (in Definition \ref{def.pseudoconvex}) defining the pseudoconvex hypersurface $\Sigma$.

The free parameter $\lambda$ plays a special role, as it allows one to absorb any lower-order terms arising from the right-hand side into the left-hand side.
To obtain unique continuation from \eqref{eq.carleman_gen}, the very rough idea is to apply the wave equation to reduce the right-hand side of \eqref{eq.carleman_gen} into lower-order quantities, which can then be absorbed into the left.
Afterwards, one can then conclude that $u = 0$ by letting $\lambda \nearrow \infty$.
See the end of Section \ref{sec.proof} for a demonstration of this process.

\begin{remark}
In contrast, for elliptic equations, every hypersurface $\Sigma$ is trivially pseudoconvex.
Thus, in elliptic settings, one immediately bypasses any difficulties from pseudoconvexity.
\end{remark}

The main difficulty in our setting is that \emph{the conformal boundary of an FG-aAdS segment $( \mi{M}, g )$ fails to be pseudoconvex}, so the classical unique continuation results no longer apply.
To see this, one can observe that there exist null geodesics near the conformal boundary that asymptote toward being everywhere tangent to the boundary.
By the above intuition, we can hence view the conformal boundary as being \emph{zero-pseudoconvex}, as it just barely fails to be pseudoconvex.

\begin{remark}
Note that pseudoconvexity is a conformally invariant property, so we can realize the conformal boundary by working with the rescaled metric $\rho^2 g$.
\end{remark}

Some early results that included zero-pseudoconvex settings are found in \cite{ken_ruiz_sog:sobolev_unique, ler_robb:unique}.
However, the modern machinery for deriving Carleman estimates (and hence unique continuation) for geometric wave equations in zero-pseudoconvex settings originated from works of the author with Alexakis and Schlue \cite{alex_schl_shao:uc_inf, alex_shao:uc_global}, in the context of unique continuation of waves from null infinities of asymptotically flat spacetimes.
Similar results were later (independently) established by Petersen \cite{peters:cpt_cauchy} for compact Cauchy horizons.
Our aAdS Carleman estimates \cite{chatz_shao:uc_ads_gauge, hol_shao:uc_ads, hol_shao:uc_ads_ns, mcgill_shao:psc_aads} also fit within this framework. 

One encounters a number of additional difficulties when dealing with a zero-pseudoconvex hypersurface $\Sigma$.
In this case, a much more careful study of the geometry near $\Sigma$ is needed to find the requisite pseudoconvexity for a Carleman estimate.
Moreover, any pseudoconvexity (if exists) must necessarily degenerate toward $\Sigma$, and this leads to various vanishing or singular weights in the Carleman estimates, making their proofs far more delicate.

Perhaps the most striking contrast with the classical theory is the size of the region from which one can uniquely continue.
The classical results are local, in that one always attains unique continuation on sufficiently small neighbourhoods of any point on the pseudoconvex hypersurface $\Sigma$.
In contrast, for zero-pseudoconvex $\Sigma$, Carleman estimates (and hence unique continuation properties) may only hold along sufficiently large regions in $\Sigma$.
This feature is exclusive to zero-pseudoconvex settings, and it is particularly relevant to aAdS conformal boundaries.

\begin{remark}
There do exist stronger unique continuation results, for which one needs far weaker assumptions than pseudoconvexity; see \cite{io_kl:unique_ip, ler:unique_tchar, robb_zuil:uc_interp, tat:uc_hh, tat:uc_hh_2}.
But, these results require either additional symmetries in the solution or some partial analyticity in the differential operator.
Unfortunately, neither of these assumptions is applicable to our aAdS setting.
\end{remark}

\subsection{The Main Estimate}

The difficulties and features mentioned above apply in particular to aAdS settings.
A key problem is to find a foliation of pseudoconvex hypersurfaces near the conformal boundary.
To find the boundary regions from which one can uniquely continue solutions of wave equations, one must determine which regions can be covered by such pseudoconvex foliations.

One important finding in aAdS settings, which is a hallmark of the Carleman estimates presented here, is that \emph{whether appropriate pseudoconvex foliations exist near the conformal boundary depends on geometric properties of the boundary itself}---namely, the GNCC.

Let us now state the Carleman estimate proved in \cite{chatz_shao:uc_ads_gauge}, again restricted to vacuum settings:

\begin{theorem}[Carleman estimate for wave equations \cite{chatz_shao:uc_ads_gauge}] \label{thm.carleman}
Let $( \mi{M}, g )$ denote a vacuum FG-aAdS segment, with conformal infinity $( \mi{I}, \gm )$.
Fix $\mi{D} \subset \mi{I}$, with $\bar{\mi{D}}$ compact, and assume $( \mi{D}, \gm )$ satisfies the GNCC, with the associated function denoted by $\eta$ (see Definition \ref{def.gncc}).
Moreover, let
\begin{equation}
\label{eq.f} f = \frac{ \rho }{ \eta } \text{,}
\end{equation}
and define, for any $f_\star > 0$, the bulk regions
\begin{equation}
\label{eq.Omega} \Omega_{ f_\star } := \{ f < f_\star \} \subset \mi{M} \text{.}
\end{equation}

Then, the following holds for any vertical tensor field $\ms{\Phi}$ on $\mi{M}$ with $\ms{\Phi}$, $\nabla \ms{\Phi}$ vanishing on $\{ f = f_\star \}$,
\begin{align}
\label{eq.carleman} &\int_{ \Omega_{ f_\star } } e^{ -\frac{2 \lambda f^p}{p} } f^{ n-2-p-2\kappa } | ( \Boxm_g + \sigma ) \ms{\Phi} |^2 \, dg \\
\notag &\qquad + \lambda^3 \limsup_{ \rho_\star \searrow 0 } \int_{ \Omega_{ f_\star } \cap \{ \rho = \rho_\star \} } ( | \Dv_\rho ( \rho^{-\kappa} \ms{\Phi} ) |^2 + | \Dv ( \rho^{-\kappa} \ms{\Phi} ) |^2 + | \rho^{-\kappa-1} \ms{\Phi} |^2 ) \, d \gv \\
\notag &\quad \gtrsim \lambda \int_{ \Omega_{ f_\star } } e^{ -\frac{2 \lambda f^p}{p} } f^{ n-2-2\kappa } ( f \rho^3 | \Dv_\rho \ms{\Phi} |^2 + f \rho^3 | \Dv \ms{\Phi} |^2 + f^{2p} | \ms{\Phi} |^2 ) \, dg \text{,}
\end{align}
provided $\kappa$ and $\lambda$ are sufficiently large (depending on $\gv$, $\mi{D}$, $\sigma$, and the rank of $\ms{\Phi}$), $f_\star$ is sufficiently small (also depending on $\gv$, $\mi{D}$, $\sigma$, and the rank of $\ms{\Phi}$), and $0 < p < \frac{1}{2}$.
\end{theorem}

The statement of Theorem \ref{thm.carleman} is quite involved, so a number of remarks are in order:

\begin{remark}
$\Boxm_g$ is the wave operator defined on vertical tensor fields; see Section \ref{sec.aads}.
Moreover, the norm $| \cdot |$ can be defined relative to a given Riemannian metric on the space of vertical tensors.
\end{remark}

\begin{remark}
The role of the Klein-Gordon mass $\sigma$ in \eqref{eq.carleman} is not superfluous.
In particular, if we rewrite the wave operator with respect to $\rho^2 g$, where the conformal boundary is finite, then $\Box_g + \sigma$ is equivalent to a wave operator $\Box_{ \rho^2 g } + c_\sigma \rho^{-2} + \text{l.o.t.}$ with a singular potential, where $c_\sigma$ is a constant depending on $\sigma$.
Furthermore, this singular potential is ``critical", in the sense that it scales in the same way as $\Box_{ \rho^2 g }$.
As a result, one must think of $\sigma$ as a principal part of the wave operator.

The parameter $\sigma$ also has a significant effect on the boundary asymptotics of solutions $u$ to
\[
( \Box_g + \sigma + \text{l.o.t.} ) u = 0 \text{.}
\]
In particular, $u$ near the conformal boundary will behave like specific powers of $\rho$, with the exponents determined by $\sigma$; for details, see, for instance, discussions in \cite{breit_freedm:stability_sgrav, hol_shao:uc_ads}.
\end{remark}

\begin{remark}
In applications of Theorem \ref{thm.carleman}, the boundary integral over $\Omega_{ f_\star } \cap \{ \rho = 0+ \}$ can be disregarded, as the Carleman estimate is generally applied to quantities that vanish to sufficiently high order at the conformal boundary.
Thus, the estimate \eqref{eq.carleman} takes the same qualitative form as \eqref{eq.carleman_gen}, save for a few additional weights depending on $f$ and $\rho$.
\end{remark}

Central to Theorem \ref{thm.carleman} and its proof is the function $f$ defined in \eqref{eq.f}.
Note that the level sets $\{ f = f_\ast \}$ of $f$, for $f_\ast > 0$ sufficiently small, form timelike hypersurface that ``hover over" the region $\mi{D}$.
Furthermore, by the assumptions on $\eta$ in Definition \ref{def.gncc}, all these level sets of $f$ terminate at the conformal boundary precisely on $\partial \mi{D}$.
Qualitatively speaking, each level set starts from the boundary at $\partial \mi{D}$ and travels inward into the bulk as one goes into $\mi{D}$.

The primary observation behind the proof of Theorem \ref{thm.carleman} is that \emph{the level sets of $f$ which are sufficiently near the conformal boundary are pseudoconvex}; this will be further elaborated in the following subsection.
We note here that the pseudoconvexity of these level sets is a consequence of $\mi{D}$ being ``sufficiently large" as to satisfy the GNCC.
Moreover, the level sets $\{ f = f_\ast \}$ asymptote precisely to the boundary region $\mi{D}$ as $f_\star \searrow 0$.
As a result, the pseudoconvexity of these level sets degenerates toward the conformal boundary, which causes considerable complications.

There is one additional difference between the forms of \eqref{eq.carleman_gen} and \eqref{eq.carleman}---the presence of additional weights in \eqref{eq.carleman_gen}, in the form of powers of $\rho$ and $f$.
The most dramatic weight $f \rho^3$, in the first-order terms on the right-hand side of \eqref{eq.carleman}, is present because of the degeneration of pseudoconvexity toward the conformal boundary.
On the other hand, the relatively mild weight $f^{2p}$ in the zero-order terms come from $\sigma$ and its effects on the boundary asymptotics of waves.

Together, these weights determine the class of wave equations for which one can obtain unique continuation results.
In particular, observe that for the left-hand side of \eqref{eq.carleman} to absorb into the right-hand side, $\ms{\Phi}$ must satisfy a Klein-Gordon equation in which the lower-order terms have weights that at most match those in the right-hand side of \eqref{eq.carleman}.
More specifically, we can only prove unique continuation for Klein-Gordon equations of the form
\begin{equation}
\label{eq.wave_valid} ( \Box_g + \sigma ) \ms{\Phi} = \mc{O} ( \rho^{ 2 + \epsilon } ) \, ( \Dv_\rho \ms{\Phi}, \Dv \ms{\Phi} ) + \mc{O} ( \rho^\epsilon ) \, \ms{\Phi} \text{,} \qquad \epsilon > 0 \text{.}
\end{equation}
Thus, another challenge in the proof of Theorem \ref{thm.correspondence} is to ensure that the relevant wave equations arising from the EVE are indeed of the above form, so that Theorem \ref{thm.carleman} is applicable.

\begin{remark}
To obtain \eqref{eq.wave_valid}, we replaced $f$-weights in \eqref{eq.carleman} by $\rho$-weights, since \eqref{eq.f} implies that $\rho$ is bounded by $f$, and because $\rho$ is a much more physically meaningful quantity.
\end{remark}

\begin{remark}
The arbitrarily small powers $\rho^\epsilon$ in the right-hand side of \eqref{eq.wave_valid} can be removed, at the cost of requiring additional orders of vanishing for $\ms{\Phi}$; see \cite{chatz_shao:uc_ads_gauge, hol_shao:uc_ads_ns, mcgill_shao:psc_aads} for details.
\end{remark}

\begin{remark}
For the weight $f^{ n - 2 - 2 \kappa }$ in \eqref{eq.carleman}, note $f^n$ counters the singular factor $\rho^{-n}$ hidden in the volume form $d g$, while $f^{ - 2 \kappa - 2 }$ matches the $\rho$-weight applied to $\ms{\Phi}$ in the boundary term.
\end{remark}

\subsection{Pseudoconvexity}

Finally, we close this section with a discussion of some of the main ideas in the proof of Theorem \ref{thm.carleman}.
Here, we focus primarily on derivation of pseudoconvexity for the level sets of $f$ and its connection to the GNCC.
For full details, the reader is referred to \cite[Section 5]{chatz_shao:uc_ads_gauge}.

According to Definition \ref{def.pseudoconvex}, we must show, in our aAdS setting, that \eqref{eq.pseudoconvex} holds whenever \eqref{eq.pseudoconvex_ass} holds.
Note that this would immediately follow if we find a smooth (scalar) function $w$ such that
\begin{equation}
\label{eq.pseudoconvex_goal} - ( \nabla^2 f + w \cdot g ) ( X, X ) > 0 \text{,}
\end{equation}
for any vector field $X$ (not necessarily null) that is tangent to the level sets of $f$.
For this, we look at the leading-order asymptotics of \eqref{eq.pseudoconvex_goal}, in terms of properties on the conformal boundary.

First, one can naturally relate tangent vectors on the conformal boundary $\mi{I}$ to tangent vectors to level sets of $f$ through the following isomorphism:
\begin{equation}
\label{eq.frame_isom} \mf{X} \in T \mi{I} \leftrightarrow X := \rho [ f ( \mf{X} \eta) \partial_\rho + \mf{X} ] \in T \{ f = f_0 \} \text{.}
\end{equation}
(The weight $\rho$ on the right-hand is to ensure that $\mf{X}$ and $X$ have almost the same $\gm$- and $g$-lengths, respectively.)
Then, a direct computation---see \cite[Lemma 2.18]{chatz_shao:uc_ads_gauge}---yields that
\[
( \nabla^2 f + f \cdot g ) ( X, X ) = \rho f^2 \Dv^2 \eta ( \mf{X}, \mf{X} ) + \frac{1}{2} \rho f \mi{L}_\rho \gv ( \mf{X}, \mf{X} ) + \rho f^3 \mf{X} \eta \mi{L}_\rho \gv ( \Dv^\sharp \eta, \mf{X} ) \text{,}
\]
where $\Dv^2$ and $\Dv^\sharp$ denote the Hessian and gradient with respect to $\gv$.
The right-hand side can then be expanded using Definition \ref{def.aads}; extracting the leading-order term at the boundary yields
\begin{align*}
- ( \nabla^2 f + f \cdot g ) ( X, X ) &= \rho f^2 ( \Dm^2 \eta - \eta \cdot \gb{2} ) ( \mf{X}, \mf{X} ) + \mc{O} ( \rho f^3 ) ( \mf{X}, \mf{X} ) \\
&= \rho f^2 ( \Dm^2 \eta + \eta \cdot \mc{P} [ \gm ] ) ( \mf{X}, \mf{X} ) + \mc{O} ( \rho f^3 ) ( \mf{X}, \mf{X} ) \text{,}
\end{align*}
where we applied the property \eqref{eq.fg_schouten} from the partial FG expansion the last step.
Choosing
\begin{equation}
\label{eq.w} w := f + f \rho^2 \zeta \text{,}
\end{equation}
for some $\zeta \in C^\infty ( \mi{I} )$ to be determined, we then have
\begin{equation}
\label{eq.pseudoconvex_pre} - ( \nabla^2 f + w \cdot g ) ( X, X ) = \rho f^2 ( \Dm^2 \eta + \eta \cdot \mc{P} [ \gm ] - \zeta \cdot \gm ) ( \mf{X}, \mf{X} ) + \mc{O} ( \rho f^3 ) ( \mf{X}, \mf{X} ) \text{.}
\end{equation}

The above can be directly linked to the GNCC through the following algebraic property:

\begin{proposition} \label{thm.hard}
The following property holds for any $\gm$-null $\mf{Z}$,
\[
( \Dm^2 \eta + \eta \cdot \mc{P} [ \gm ] ) ( \mf{Z}, \mf{Z} ) \gtrsim \mf{h} ( \mf{Z}, \mf{Z} ) \text{,}
\]
if and only if there exists $\zeta \in C^\infty ( \mi{I} )$ such that the following holds for any vector field $\mf{X}$ on $\mi{I}$:
\[
( \Dm^2 \eta + \eta \cdot \mc{P} [ \gm ] - \zeta \cdot \gm ) ( \mf{X}, \mf{X} ) \gtrsim \mf{h} ( \mf{X}, \mf{X} ) \text{.}
\]
\end{proposition}

\begin{remark}
See \cite[Corollary 3.5]{mcgill_shao:psc_aads} or \cite[Lemma 4.3]{tat:notes_uc} for proofs of Proposition \ref{thm.hard}.
Note the $n = 2$ case can be proved directly, but when $n \geq 3$, the main (rather nontrivial) observation is that two bilinear forms that do not vanish simultaneously can be simultaneously diagonalized \cite{greub:lin_alg}.
\end{remark}

Combining \eqref{eq.pseudoconvex_pre} with Proposition \ref{thm.hard}, we conclude that \emph{if the GNCC holds, then \eqref{eq.pseudoconvex_goal} holds sufficiently near the conformal boundary, and the near-boundary level sets of $f$ are pseudoconvex}.

\begin{remark}
In fact, the above computations give a slightly stronger bound:
\[
- ( \nabla^2 f + w \cdot g ) ( X, X ) \gtrsim \rho f^2 \cdot \mf{h} ( \mf{X}, \mf{X} ) \text{.}
\]
The specific weights on the right-hand side eventually lead to corresponding weights in \eqref{eq.carleman}.
\end{remark}

As for the Carleman estimate \eqref{eq.carleman}, the process is similar to other proofs of Carleman estimates for wave equations.
Conjugating the Klein-Gordon operator $( \Box_g + \sigma )$ with the Carleman weight,
\[
\mc{L} := e^{ - \frac{ \lambda f^p }{p} } f^{ -\kappa } ( \Boxm_g + \sigma ) e^{ \frac{ \lambda f^p }{p} } f^\kappa \text{,}
\]
we can then roughly view the proof of \eqref{eq.carleman} (modulo some additional weights) as a positive commutator estimate for $\mc{L}$.
The challenge is to obtain positivity for the $H^1$-norm of $\ms{\Phi}$:
\begin{itemize}
\item For derivatives of $\ms{\Phi}$ along the normal to the level sets of $f$, the positivity comes, as usual, for free from the square of the odd part of $\mc{L}$.

\item For $\ms{\Phi}$ itself, the positivity comes from the choice of Carleman weight---in fact, from choosing the parameter $\kappa$ to be sufficiently large compared to $\sigma$.

\item For derivatives of $\ms{\Phi}$ along the level sets of $f$, we look, as usual, to the commutator between the odd and even parts of $\mc{L}$.
For this, the principal part is roughly given by
\[
- f^{n-3} ( \nabla^2 f + w \cdot g ) ( \nablam \ms{\Phi}, \nablam \ms{\Phi} ) \text{,}
\]
which we can show is positive due to the GNCC and pseudoconvexity.
\end{itemize}
The proof of \eqref{eq.carleman} is completed (after a substantial amount of computations) by using the above observations, and by noting that all remaining quantities are ``lower-order" and can be absorbed.

\begin{remark}
The proof of \eqref{eq.carleman} in \cite{chatz_shao:uc_ads_gauge} does not use a positive commutator argument, but instead derives positivity directly through geometric computations via integrations by parts.
We adopted the language of positive commutators here, since it is widely used in the Carleman estimate literature.
\end{remark}

Finally, all the above is significantly complicated by the zero-pseudoconvexity, which introduces various weights that vanish at the conformal boundary.
In particular, considerable care is needed to ensure that terms which should be ``lower-order" can still be absorbed into the ``principal" terms---that is, that the various weights are compatible with each other.
For details of these computations, the reader is referred to the full proof of Theorem \ref{thm.carleman} in \cite[Section 5]{chatz_shao:uc_ads_gauge}.

\section{Proof of Theorem \ref{thm.correspondence}} \label{sec.proof}

In this section, we give an outline of the proof of our main result, Theorem \ref{thm.correspondence}, and we highlight the main ideas behind the proof.
The reader is referred to \cite{hol_shao:uc_ads_eve} for further details.

\subsection{Reduction to Unique Continuation}

Let us now assume the hypotheses of Theorem \ref{thm.correspondence}.
As a matter of convention, we use the same notations for quantities defined from both $g$ and $\bar{g}$, except that objects associated with $\bar{g}$ have a bar over their symbols.

The initial step is to reduce Theorem \ref{thm.correspondence} to a setting in which can perform analysis.
By pulling $\bar{\mi{I}}$ and $\bar{\mi{M}} := ( 0, \bar{\rho}_0 ] \times \bar{\mi{I}}$ back through the gauge transformation $\phi$, we are then in the setting in which $\bar{\mi{I}} = \mi{I}$ and $\bar{\mi{M}} = \mi{M}$---more specifically, $\bar{\mi{M}}$ is pulled back to $\mi{M}$ via
\[
( \rho, p ) \in \mi{M} \mapsto ( \rho, \phi (p) ) \in \bar{\mi{M}} \text{.}
\]
Furthermore, $( \gb{0}, \gb{n} )$ and $( \gbc{0}, \gbc{n} )$ are now related on $\mi{D}$ by \eqref{eq.fg_conformal}--\eqref{eq.fg_change}, for some conformal factor $\mf{a}$.
Next, applying a change of FG boundary defining function $\rho \mapsto \bar{\rho}$ satisfying the boundary condition \eqref{eq.fg_rho}, we can then further reduce to the case in which
\begin{equation}
\label{eq.simple} \mi{I} = \bar{\mi{I}} \text{,} \qquad \mi{M} = \bar{\mi{M}} \text{,} \qquad ( \gb{0}, \gb{n} ) |_{ \mi{D} } = ( \gbc{0}, \gbc{n} ) |_{ \mi{D} } \text{.}
\end{equation}

Now that $g$ and $\bar{g}$---as well as all the quantities defined from them---lie on a common manifold, the proof will be complete if we can show that $g = \bar{g}$, or equivalently, $g - \bar{g} = 0$.
By the last part of \eqref{eq.simple}, and by the properties of the partial FG expansion in Theorem \ref{thm.aads_fg}, we have that
\begin{equation}
\label{eq.fg_match} \gb{k} - \gbc{k} = 0 \text{,} \qquad \gb{\star} - \gbc{\star} = 0 \text{,} \qquad 0 \leq k \leq n \text{.}
\end{equation}
Thus, thanks to \eqref{eq.aads}, Theorem \ref{thm.correspondence} will be proved if we show that \eqref{eq.fg_match} implies
\begin{equation}
\label{eq.goal} \gv - \bar{\gv} = 0 \text{.}
\end{equation}

\subsection{The Wave-Transport System}

The strategy for proving \eqref{eq.goal} will be to formulate $\gv - \bar{\gv}$ as an unknown in a closed system of (vertical) tensorial transport and wave equations, and to apply our Carleman estimates to this system to conclude that all its unknowns (including $\gv - \bar{\gv}$) vanish.

In this subsection, we derive a viable wave-transport system.
For this, we also define
\begin{equation}
\label{eq.Lv} \Lv_{ab} := \mi{L}_\rho \gv_{ab} \text{.}
\end{equation}
In addition, we let $W$ denote the Weyl curvature for $g$, and we decompose $W$ into vertical quantities:
\begin{equation}
\label{eq.wv} \wv^0_{ a b c d } := \rho^2 \, W_{ a b c d } \text{,} \qquad \wv^1_{ a b c } := \rho^2 \, W_{ \rho a b c } \text{,} \qquad \wv^2_{ a b } := \rho^2 \, W_{ \rho a \rho b } \text{,}
\end{equation}
The unknowns for our system will be constructed from $\gv$, \eqref{eq.Lv}, and \eqref{eq.wv}.

\begin{remark}
We use lowercase Latin indices for vertical components, i.e., in directions along $\mi{I}$.
\end{remark}

From the Gauss-Codazzi equations on level sets of $\rho$ and \eqref{eq.eve}, one derives the following:
\begin{equation}
\label{eq.transport_init} \mi{L}_\rho \Lv_{ab} = -2 \wv^2_{ab} + \rho^{-1} \Lv_{ab} + \frac{1}{2} \gv^{cd} \Lv_{ad} \Lv_{bc} \text{,} \qquad \Dv_b \Lv_{ a c } - \Dv_a \Lv_{ b c } = 2 \wv^1_{ c a b } \text{.}
\end{equation}
Of course, analogous formulas also hold with respect to quantities constructed from $\bar{\gv}$.
We wish to couple the equations \eqref{eq.transport_init} to wave equations satisfied by the Weyl curvature $W$:
\begin{equation}
\label{eq.wave_spt} \Box_g W + 2n W = W \cdot W \text{.}
\end{equation}
(The right-hand side is quadratic in $W$; see \cite[Proposition 3.4]{hol_shao:uc_ads_eve} or \cite[Proposition 4.5]{mcgill:loc_ads} for the exact formula or a derivation.)
For this, we define one additional renormalization for $\wv^0$:
\begin{equation}
\label{eq.wstar} \wv^\star_{ a b c d } := \wv^0_{ a b c d } - \frac{1}{ n - 2 } ( \gv_{ a d } \wv^2_{ b c } + \gv_{ b c } \wv^2_{ a d } - \gv_{ a c } \wv^2_{ b d } - \gv_{ b d } \wv^2_{ a c } ) \text{.} 
\end{equation}
Then, one obtains, from \eqref{eq.wave_spt}, the following tensorial wave equations for $\wv^2$, $\wv^1$, and $\wv^\star$ (see either \cite[Proposition 3.6]{hol_shao:uc_ads_eve} or \cite[Proposition 4.15]{mcgill:loc_ads} for detailed derivations):
\begin{align}
\label{eq.wave_init} \Boxm_g \wv^2 + 2 ( n - 2 ) \wv^2 &= \ms{N}^2 ( \gv, \Lv, \Dv \Lv, \wv^\star, \wv^1, \wv^2, \Dv \wv^\star, \Dv \wv^1, \Dv \wv^2 ) \text{,} \\
\notag \Boxm_g \wv^1 + ( n - 1 ) \wv^2 &= \ms{N}^1 ( \gv, \Lv, \Dv \Lv, \wv^\star, \wv^1, \wv^2, \Dv \wv^\star, \Dv \wv^1, \Dv \wv^2 ) \text{,} \\
\notag \Boxm_g \wv^\star &= \ms{N}^\star ( \gv, \Lv, \Dv \Lv, \wv^\star, \wv^1, \wv^2, \Dv \wv^\star, \Dv \wv^1, \Dv \wv^2 ) \text{.}
\end{align}
Here, $\ms{N}^2 (\cdot)$, $\ms{N}^1 (\cdot)$, $\ms{N}^\star (\cdot)$ represent nonlinear terms involving (contractions of) the listed quantities that decay sufficiently quickly toward the conformal boundary.

\begin{remark}
The power $\rho^2$ in \eqref{eq.wv} is crucial, as this is the only power of $\rho$ for which the resulting quantities will satisfy wave equations that are compatible with \eqref{eq.carleman}.
The renormalization \eqref{eq.wstar} is made for the same reason, as the wave equation for $\wv^0$ cannot be used with \eqref{eq.carleman}.
\end{remark}

\begin{remark}
That different Klein-Gordon masses appear in \eqref{eq.wave_init}, at least when $n > 3$, is due to $\wv^\star$, $\wv^1$, and $\wv^2$ having different asymptotic behaviors at the conformal boundary.
(The case $n = 3$ is exceptional, since $\wv^1$ and $\wv^2$ have the same mass, and $\wv^\star$ is fully determined by $\wv^1$ and $\wv^2$.)
As a result, we must treat the components $\wv^\star$, $\wv^1$, $\wv^2$ separately in our analysis.
\end{remark}

Subtracting \eqref{eq.Lv}, \eqref{eq.transport_init}, and \eqref{eq.wave_init} from their counterparts for $\bar{\gv}$ yields a closed wave-transport system for the quantities $\gv - \bar{\gv}$,  $\Lv - \bar{\Lv}$, $\wv^\star - \bar{\wv}^\star$, $\wv^1 - \bar{\wv}^1$, and $\wv^2 - \bar{\wv}^2$.
However, \emph{this system fails to close for the purpose of applying our Carleman estimates}.
In particular, the wave Carleman estimate \eqref{eq.carleman} will only allow us to control up to one derivative of $\wv^\star - \bar{\wv}^\star$, $\wv^1 - \bar{\wv}^1$, and $\wv^2 - \bar{\wv}^2$.
This, in turn, only allows to control---through \eqref{eq.Lv} and \eqref{eq.transport_init}---one derivative of $\Lv - \bar{\Lv}$ and $\gv - \bar{\gv}$.
On the other hand, when we take differences of the wave equations \eqref{eq.wave_init}, we obtain terms involving the difference of $\smash{\Boxm}$ and $\smash{\bar{\Boxm}}$, which contains second derivatives of $\gv - \bar{\gv}$ that we a priori cannot treat.

The resolution, inspired by the symmetry extension result \cite{io_kl:killing} of Ionescu and Klainerman, is to apply a careful renormalization of the system that eliminates the troublesome quantities.

The first crucial observation is that while $\Dv^2 ( \gv - \bar{\gv} )$ is off limits, we can obtain some control if one derivative is a curl.
In particular, \eqref{eq.Lv} and the second part of \eqref{eq.transport_init} yield, schematically,
\begin{align*}
\mi{L}_\rho [ \Dv_{db} ( \gv - \bar{\gv} )_{ac} - \Dv_{da} ( \gv - \bar{\gv} )_{bc} ] &\sim \Dv_d [ ( \Dv_b \Lv_{ a c } - \Dv_a \Lv_{ b c } ) - ( \bar{\Dv}_b \bar{\Lv}_{ a c } - \bar{\Dv}_a \bar{\Lv}_{ b c } ) ] \\
&\sim \Dv_d ( \wv^1 - \bar{\wv}^1 )_{cab} \text{.}
\end{align*}
The above still does not quite suffice, and we need one more renormalization---this is due to error terms on the right-hand side involving $\Dv ( \check{\Dv} - \Dv )$, which again contain the undesirable $\Dv^2 ( \gv - \check{\gv} )$.
All this leads us to define the following auxiliary quantities:
\begin{align}
\label{eq.intro_proof_B} \ms{B}_{cab} &:= \Dv_c ( \gv - \bar{\gv} )_{ab} - \Dv_a ( \gv - \bar{\gv} )_{cb} - \Dv_b \ms{Q}_{ca} \text{,} \\
\notag \mi{L}_\rho \ms{Q}_{ca} &:= \gv^{de} \Lv_{ce} ( \gv - \bar{\gv} )_{ad} - \gv^{de} \Lv_{ae} ( \gv - \bar{\gv} )_{cd} \text{,}
\end{align}
with $\ms{Q} \rightarrow 0$ as $\rho \rightarrow 0$.
One then shows $\Dv \ms{B}$ can indeed be adequately controlled by $\Dv ( \wv^1 - \bar{\wv}^1 )$.

The second crucial observation comes from a detailed examination of the difference $\smash{ \Boxm - \bar{\Boxm} }$.
To appreciate this, let us consider the wave equation for $\wv^2 - \bar{\wv}^2$:
\begin{equation}
\label{eq.box_diff_1} \Boxm ( \wv^2 - \bar{\wv}^2 ) = \Boxm \wv^2 - \bar{\Boxm} \bar{\wv}^2 - ( \Boxm - \bar{\Boxm} ) \bar{\wv}^2 \text{.}
\end{equation}
The dangerous terms arise from the following (rather long) computation,
\begin{align}
\label{eq.box_diff_2} ( \Boxm - \bar{\Boxm} ) \bar{\wv}^2_{ab} &= - \frac{1}{2} \rho^2 \gv^{cd} \gv^{ef} \bar{\wv}^2_{eb} \Dv_c \ms{B}_{afd} - \frac{1}{2} \gv^{ef} \bar{\wv}^2_{eb} \Boxm ( \gv - \bar{\gv} + \ms{Q} )_{af} \\
\notag &\qquad - \frac{1}{2} \rho^2 \gv^{cd} \gv^{ef} \bar{\wv}^2_{ea} \Dv_c \ms{B}_{bfd} - \frac{1}{2} \gv^{ef} \bar{\wv}^2_{ea} \Boxm ( \gv - \bar{\gv} + \ms{Q} )_{bf} + \ms{Err}_{ab} \text{,}
\end{align}
where $\ms{Err}$ consists of (many) terms containing only difference quantities that we can control.

The key point is that the only instances of $\Dv^2 ( \gv - \check{\gv} )$ appear either as $\Dv \ms{B}$, which we can control, or as $\smash{\Boxm}$ applied to difference quantities.
This leads us to the renormalized curvature difference
\begin{equation}
\label{eq.W} \ms{W}^2_{ab} := \wv^2_{ab} - \bar{\wv}^2_{ab} + \frac{1}{2} \gv^{de} \bar{\wv}^2_{ad} ( \gv - \bar{\gv} + \ms{Q} )_{be} + \frac{1}{2} \gv^{de} \bar{\wv}^2_{db} ( \gv - \bar{\gv} + \ms{Q} )_{ae} \text{,}
\end{equation}
which in essence shifts the $\smash{\Boxm}$-terms from the right-hand side of \eqref{eq.box_diff_2} into the left; one can also define the remaining $\ms{W}^1$ and $\ms{W}^\star$ similarly.
In light of \eqref{eq.box_diff_1} and \eqref{eq.box_diff_2}, we obtain that $\ms{W}^\star$, $\ms{W}^1$, $\ms{W}^2$ satisfy wave equations that do not contain $\Dv^2 ( \gv - \check{\gv} )$ as sources.

Finally, the renormalized wave-transport system is obtained by treating the quantities
\begin{equation}
\label{eq.unknowns} \gv - \bar{\gv} \text{,} \qquad \ms{Q} \text{,} \qquad \Lv - \bar{\Lv} \text{,} \qquad \ms{B} \text{,} \qquad \ms{W}^\star \text{,} \qquad \ms{W}^1 \text{,} \qquad \ms{W}^2
\end{equation}
as unknowns.
In particular, from the above discussions, and from various asymptotic properties of geometric quantities, we arrive at the (schematic) transport equations
\begin{align}
\label{eq.transport} \mi{L}_\rho ( \gv - \bar{\gv} ) &= \Lv - \bar{\Lv} \text{,} \\
\notag \mi{L}_\rho \ms{Q} &= \mc{O} (\rho) \, ( \gv - \bar{\gv}, \ms{Q} ) \text{,} \\
\notag \mi{L}_\rho \ms{B} &= 2 ( \wv^1 - \bar{\wv}^1 ) + \mc{O} (\rho) \, ( \gv - \bar{\gv}, \ms{Q}, \ms{B} ) + \mc{O} (1) \, ( \Lv - \bar{\Lv} ) \text{,} \\
\notag \mi{L}_\rho [ \rho^{-1} ( \Lv - \bar{\Lv} ) ] &= - 2 \rho^{-1} \ms{W}^2 + \mc{O} (1) \, ( \gv - \bar{\gv}, \Lv - \bar{\Lv}, \ms{Q} ) \text{,}
\end{align}
coupled to the following (schematic) wave equations:
\begin{align}
\label{eq.wave} \left. \begin{matrix} \Boxm \ms{W}^2 + 2 ( n - 2 ) \ms{W}^2 \\ \Boxm \ms{W}^1 + ( n - 1 ) \ms{W}^1 \\ \Boxm \ms{W}^\star \end{matrix} \right\} &= \sum_{ \ms{V} \in \{ {\ms{W}^\star}, \ms{W}^1, \ms{W}^2 \} } \left[ \mc{O} (\rho^2) \ms{V} + \mc{O} (\rho^3) \Dv \ms{V} \right] + \mc{O} (\rho) \, ( \Lv - \bar{\Lv} ) \\
\notag &\qquad + \mc{O} (\rho^2) \, ( \gv - \bar{\gv}, \ms{Q}, \Dv ( \gv - \bar{\gv} ), \Dv \ms{Q}, \Dv ( \Lv - \bar{\Lv} ), \Dv \ms{B} ) \text{.}
\end{align}
The $\mc{O} ( \cdot )$'s in \eqref{eq.transport}--\eqref{eq.wave} indicate the asymptotics of various coefficients as $\rho \rightarrow 0$.

For more precise formulas and derivations, see \cite[Proposition 3.13, Proposition 3.14]{hol_shao:uc_ads_eve}.
In particular, the wave-transport system \eqref{eq.transport}--\eqref{eq.wave} indeed closes from the point of view of derivatives.

\begin{remark}
The system \eqref{eq.transport}--\eqref{eq.wave} could also be used to derive unique continuation results for the EVE near finite timelike hypersurfaces, providing an alternative approach to that of \cite{alex:uc_vacuum}.
\end{remark}

\begin{remark}
Note that \cite{biq:uc_einstein, chru_delay:uc_killing}, the Riemannian and stationary analogues of Theorem \ref{thm.correspondence}, avoid the issues that make the extended system \eqref{eq.transport}--\eqref{eq.wave} necessary.
In particular, the main equations in \cite{biq:uc_einstein, chru_delay:uc_killing} are elliptic rather than hyperbolic, and the Carleman estimates in this setting control one extra derivative.
Thus, the analogue of $\Dv^2 ( \gv - \bar{\gv} )$ does not cause any issues there.

Furthermore, \cite{biq:uc_einstein, chru_delay:uc_killing} avoid the curvature entirely.
Instead, they can work with a much simpler system that is based on second-order elliptic equations for $\Lv - \bar{\Lv}$.

Finally, since all hypersurfaces are pseudoconvex in elliptic settings, \cite{biq:uc_einstein, chru_delay:uc_killing} avoid all the major difficulties of zero-pseudoconvexity and of constructing pseudoconvex hypersurfaces.
\end{remark}

\begin{remark}
As previously mentioned, our system \eqref{eq.transport}--\eqref{eq.wave} is heavily inspired by \cite{io_kl:killing}, which devised an analogous wave-transport system for proving a Killing field extension result across finite hypersurfaces.
In particular, \cite{io_kl:killing} used as unknowns the analogues of $\mi{L}_K g$, $\nabla \mi{L}_K g$, $\mi{L}_K W$, along with a number of renormalizations, where $K$ is the Killing vector field to be extended.

In fact, a modification of the wave-transport system in \cite{io_kl:killing} (that also decomposes the unknowns into vertical tensor fields) can be used to give a direct proof of Theorem \ref{thm.killing}, without appealing to Theorem \ref{thm.correspondence} or Theorem \ref{thm.symmetry}.
This modified wave-transport system would have the same qualities as \eqref{eq.transport}--\eqref{eq.wave}, one we can similarly apply our Carleman estimates to this system.

Furthermore, one can draw a direct parallel between the system in \cite{io_kl:killing} and our wave-transport system \eqref{eq.transport}--\eqref{eq.wave}.
More specifically, one can map the unknowns in \cite{io_kl:killing} to the unknowns here by replacing every $\mi{L}_K$ applied to a quantity in \cite{io_kl:killing} by the corresponding difference of that quantity for two metrics here.
(To be fully precise, one must also decompose spacetime quantities into vertical tensor fields.)
Furthermore, the renormalizations in \cite{io_kl:killing} can be related to our renormalizations in this same manner.
For details, see the discussions in \cite[Section 1.5]{hol_shao:uc_ads_eve}.
\end{remark}

\subsection{Unique Continuation}

With the wave-transport system \eqref{eq.transport}--\eqref{eq.wave} in place, we now apply our Carleman estimates to this system to derive unique continuation---in particular \eqref{eq.goal}.

The first step is to show that the unknowns of \eqref{eq.transport}--\eqref{eq.wave} vanish to high enough order at the conformal boundary.
For this, we begin by expressing them in terms of our partial FG expansion \eqref{eq.aads_fg}.
Using that the expansions for $\gv$ and $\bar{\gv}$ match up to $n$-th order by \eqref{eq.fg_match}, we derive
\begin{equation}
\label{eq.vanish_mid} \gv - \bar{\gv} = o ( \rho^n ) \text{,} \qquad \Lv - \bar{\Lv}, \wv^1 - \bar{\wv}^1 = o ( \rho^{n-1} ) \text{,} \qquad \wv^2 - \bar{\wv}^2, \wv^0 - \bar{\wv}^0 = o ( \rho^{n-2} ) \text{.}
\end{equation}

From here, the key observation is the following system (see \cite[Proposition 3.15]{hol_shao:uc_ads_eve} for details):
\begin{align}
\label{eq.transport_2} \mi{L}_\rho ( \gv - \bar{\gv} ) &= \Lv - \bar{\Lv} \text{,} \\
\notag \mi{L}_\rho [ \rho^{-1} ( \Lv - \bar{\Lv} ) ] &= - 2 \rho^{-1} ( \wv^2 - \bar{\wv}^2 ) + \mc{O} ( \rho ) \, ( \gv - \bar{\gv} ) + \mc{O} (1) \, ( \Lv - \bar{\Lv} ) \text{,} \\
\notag \mi{L}_\rho [ \rho^{2-n} ( \wv^2 - \bar{\wv}^2 ) ] &= \mc{O} ( \rho^{2-n} ) \, ( \Dv ( \wv^1 - \bar{\wv}^1 ), \Lv - \bar{\Lv} ) \\
\notag &\qquad + \mc{O} ( \rho^{3-n} ) \, ( \gv - \bar{\gv}, \Dv ( \gv - \bar{\gv} ), \wv^0 - \bar{\wv}^0, \wv^2 - \bar{\wv}^2 ) \text{,} \\
\notag \mi{L}_\rho [ \rho^{-1} ( \wv^1 - \bar{\wv}^1 ) ] &= \mc{O} ( \rho^{-1} ) \, ( \Dv ( \wv^2 - \bar{\wv}^2 ), \Lv - \bar{\Lv} ) \\
\notag &\qquad + \mc{O} (1) \, ( \gv - \bar{\gv}, \Dv ( \gv - \bar{\gv} ), \wv - \bar{\wv}^1 ) \text{,} \\
\notag \mi{L}_\rho ( \wv^0 - \bar{\wv}^0 ) &= \mc{O} ( \rho^{-1} ) \, ( \wv^2 - \bar{\wv}^2 ) + \mc{O} (1) \, ( \Dv ( \wv^1 - \bar{\wv}^1 ), \gv - \bar{\gv}, \Lv - \bar{\Lv} ) \\
\notag &\qquad + \mc{O} ( \rho ) \, ( \Dv ( \gv - \bar{\gv} ), \wv^0 - \bar{\wv}^0 ) \text{.} 
\end{align}
Note that the first two parts of \eqref{eq.transport_2} are just \eqref{eq.Lv} and \eqref{eq.transport} (without renormalizations), while the last three parts of \eqref{eq.transport_2} are consequences of the Bianchi equations for $W$.
Integrating \eqref{eq.transport_2} from $\rho = 0$, one can improve the vanishing in \eqref{eq.vanish_mid} by two powers of $\rho$:
\begin{equation}
\label{eq.vanishing_high} \gv - \bar{\gv} = o ( \rho^{n+2} ) \text{,} \qquad \Lv - \bar{\Lv}, \wv^1 - \bar{\wv}^1 = o ( \rho^{n+1} ) \text{,} \qquad \wv^2 - \bar{\wv}^2, \wv^0 - \bar{\wv}^0 = o ( \rho^n ) \text{.}
\end{equation}

\begin{remark}
The initial vanishing \eqref{eq.vanish_mid} from the partial FG expansions is crucial here.
Due to the $\rho^{2-n}$ in the third part of \eqref{eq.transport_2}, we cannot use \eqref{eq.transport_2} to improve our vanishing without \eqref{eq.vanish_mid}.
\end{remark}

\begin{remark}
The equations \eqref{eq.transport_2} have a hierarchical structure, in that we have to apply them in a specific order to correctly derive \eqref{eq.vanishing_high}.
In particular, in \cite{hol_shao:uc_ads_eve}, the quantities are estimated in the following order:\ $\wv^2 - \bar{\wv}^2$, $\wv^0 - \bar{\wv}^0$, $\Lv - \bar{\Lv}$, $\gv - \bar{\gv}$, $\wv^1 - \bar{\wv}^1$.
\end{remark}

The vanishing rates can be further improved by repeating the above argument, so that each iteration improves all vanishing rates by $\rho^2$.
From this, we conclude that all the unknowns \eqref{eq.unknowns} vanish to arbitrarily high order at the conformal boundary.
As a result, \emph{there will be no boundary terms present when we apply the Carleman estimates}.

\begin{remark}
More accurately, we can obtain as much vanishing for \eqref{eq.unknowns} as we have regularity in the vertical directions for $\gv$.
As a result, we must also assume sufficient regularity in Theorem \ref{thm.correspondence} in order to achieve the vanishing rates needed for \eqref{eq.unknowns} so that the Carleman estimates can be applied.
However, this is mainly a technical issue, and we avoid dealing with this here.
\end{remark}

From here, the process is mostly standard.
We fix constants
\begin{equation}
\label{eq.f_const} 0 < f_i < f_e < f_\star \text{,} \qquad \kappa \gg 1 \text{,}
\end{equation}
with $f_\star$ sufficiently small and $\kappa$ large enough, and we fix a smooth cutoff function
\begin{equation}
\label{eq.cutoff} \chi := \hat{\chi} (f) \text{,} \qquad \hat{\chi} (s) = \begin{cases} 1 & s < f_i \text{,} \\ 0 & s > f_e \text{.} \end{cases}
\end{equation}

Applying the wave Carleman estimate \eqref{eq.carleman} on $\Omega_{ f_\star }$ to $\chi \ms{W}^\star$, $\chi \ms{W}^1$, and $\chi \ms{W}^2$, summing the results, and recalling the wave equations \eqref{eq.wave}, we obtain that
\begin{align}
\label{eq.uc_wave} &\lambda \int_{ \Omega_i } w_\lambda (f) \sum_{ \ms{V} \in \{ \ms{W}^\star, \ms{W}^1, \ms{W}^2 \} } ( \rho^{2p} | \ms{V} |^2 + \rho^4 | \ms{D} \ms{V} |^2 ) \, d g \\
\notag &\quad \lesssim \int_{ \Omega_i } w_\lambda (f) \sum_{ \ms{U} \in \{ \gv - \bar{\gv}, \ms{Q}, \Lv - \bar{\Lv}, \ms{B} \} } ( \rho^{ 2 - p } | \ms{U} |^2 + \rho^{ 4 - p } | \Dv \ms{U} |^2 ) \, d g \\
\notag &\quad\qquad + \int_{ \Omega_i } w_\lambda (f) \sum_{ \ms{V} \in \{ \ms{W}^\star, \ms{W}^1, \ms{W}^2 \} } ( \rho^{ 4 - p } | \ms{V} |^2 + \rho^{ 6 - p } | \Dv \ms{V} |^2 ) \, d g + \int_{ \Omega_e } w_\lambda (f) \, ( \dots ) \, d g \text{,}
\end{align}
where
\begin{equation}
\label{eq.weight} w_\lambda (f) := e^{ -2\lambda p^{-1} f^p } f^{n-2-2\kappa} \text{,} \qquad \Omega_i := \{ 0 < f_i \} \text{,} \qquad \Omega_e := \{ f_i \leq f < f_e \} \text{.}
\end{equation}
(The ``$\dots$" in the $\Omega_e$-integral depends on the unknowns \eqref{eq.unknowns}, various weights in $\rho$ and $f$, and the cutoff $\chi$.
Its precise contents are irrelevant, as we only require that this integral is finite.)

For the quantities satisfying transport equations, we must apply a separate transport Carleman estimate (see \cite[Proposition 4.6]{hol_shao:uc_ads_eve} for the precise statement and proof):

\begin{proposition}[Carleman estimate for transport equations \cite{hol_shao:uc_ads_eve}] \label{thm.carleman_transport}
Assume the same setup as Theorem \ref{thm.carleman}, and fix $s \geq 0$.
Then, the following holds for any vertical tensor field $\ms{\Psi}$ on $\mi{M}$; large enough $\kappa$, $\lambda$ (depending on $n$ and $s$); small enough $f_\star$ (depending on $\gv$ and $\mi{D}$); and $0 < p < \frac{1}{2}$:
\begin{align}
\label{eq.carleman_transport} &\int_{ \Omega_{ f_\star } } e^{ -\lambda p^{-1} f^p } f^{ n - 2 - p - 2 \kappa } \rho^{s+2} | \mi{L}_\rho \ms{\Psi} |^2 \, dg + \lambda \limsup_{ \rho_\star \searrow 0 } \int_{ \Omega_{ f_\star } \cap \{ \rho = \rho_\ast \} } \rho^s | \rho^{ - \kappa - 1 } \ms{\Psi} |^2 \, d \gv \\
\notag &\quad \gtrsim \lambda \int_{ \Omega_{ f_\star } } e^{ -\lambda p^{-1} f^p } f^{ n - 2 - 2 \kappa } \rho^s | \ms{\Psi} |^2 \, dg \text{.}
\end{align}
\end{proposition}

\begin{remark}
While the wave Carleman estimate, Theorem \ref{thm.carleman}, was quite involved and required several new ideas to treat, the transport Carleman estimate is straightforward to prove.
\end{remark}

We now apply Proposition \ref{thm.carleman_transport} to $\gv - \bar{\gv}$, $\ms{Q}$, $\Lv - \bar{\Lv}$, $\ms{B}$, and their first derivatives (with the same $\kappa$ and $f_\star$ as before, and with appropriate $s$).
Recalling \eqref{eq.transport} and summing the results yields
\begin{align}
\label{eq.uc_transport} &\lambda \int_{ \Omega_i } w_\lambda (f) \sum_{ \ms{U} \in \{ \gv - \bar{\gv}, \ms{Q}, \Lv - \bar{\Lv}, \ms{B} \} } ( | \ms{U} |^2 + \rho^3 | \Dv \ms{U} |^2 ) \, d g \\
\notag &\quad \lesssim \int_{ \Omega_i } w_\lambda (f) \sum_{ \ms{U} \in \{ \gv - \bar{\gv}, \ms{Q}, \Lv - \bar{\Lv}, \ms{B} \} } ( \rho^{ 2 - p } | \ms{U} |^2 + \rho^{ 5 - p } | \Dv \ms{U} |^2 ) \, d g \\
\notag &\quad \qquad + \int_{ \Omega_i } w_\lambda (f) \sum_{ \ms{V} \in \{ \ms{W}^\star, \ms{W}^1, \ms{W}^2 \} } ( \rho^{ 2 - p } | \ms{V} |^2 + \rho^{ 5 - p } | \Dv \ms{V} |^2 ) \, d g \text{.}
\end{align}

The key point here is that the $\rho$-weights on the right-hand sides of \eqref{eq.uc_wave} and \eqref{eq.uc_transport} come from the $\mc{O} ( \cdot )$-coefficients in \eqref{eq.transport}--\eqref{eq.wave}.
The final crucial feature of our system is these $\rho$-weights are strong enough that, after summing \eqref{eq.uc_wave} and \eqref{eq.uc_transport}, \emph{the $\Omega_i$-integrals on the right-hand side can be absorbed into the left-hand side} (once $\lambda$ is sufficiently large).
From the above, we conclude that
\[
\lambda \int_{ \Omega_i } w_\lambda (f) \sum_{ \ms{V} \in \{ \ms{W}^\star, \ms{W}^1, \ms{W}^2 \} } \rho^{2p} | \ms{V} |^2 \, d g + \lambda \int_{ \Omega_i } w_\lambda (f) \sum_{ \ms{U} \in \{ \gv - \bar{\gv}, \ms{Q}, \Lv - \bar{\Lv}, \ms{B} \} } | \ms{U} |^2 \, d g \lesssim \int_{ \Omega_e } w_\lambda (f) \, ( \dots ) \, d g \text{.}
\]
Finally, $w_\lambda (f)$ in the above can be removed in the standard fashion by noting that $w_\lambda (f) \leq w_\lambda ( f_i )$ on $\Omega_e$ and $w_\lambda (f) \geq w_\lambda ( f_i )$ on $\Omega_i$; this yields the estimate
\[
\lambda \int_{ \Omega_i } \sum_{ \ms{V} \in \{ \ms{W}^\star, \ms{W}^1, \ms{W}^2 \} } \rho^{2p} | \ms{V} |^2 \, d g + \lambda \int_{ \Omega_i } \sum_{ \ms{U} \in \{ \gv - \bar{\gv}, \ms{Q}, \Lv - \bar{\Lv}, \ms{B} \} } | \ms{U} |^2 \, d g \lesssim \int_{ \Omega_e } ( \dots ) \, d g \text{.}
\]
Letting $\lambda \nearrow \infty$ in the above, we conclude that \eqref{eq.goal} holds on $\Omega_i$.

\begin{remark}
By setting $f_i$ to be arbitrarily close to $f_\star$, we hence show that \eqref{eq.goal} holds on $\Omega_{ f_\star }$.
\end{remark}

\section{Related Problems} \label{sec.open}

We conclude this article by briefly discussing some other research directions related to Theorem \ref{thm.correspondence}.
Much of this discussion can also be found in \cite[Section 1.6]{hol_shao:uc_ads_eve}.

\subsection{Low Dimensions}

Recall that Theorems \ref{thm.correspondence}, \ref{thm.symmetry}, and \ref{thm.killing} all assume that the conformal boundary dimension $n$ is strictly greater than $2$.
This raises the question of whether analogues of these theorems hold in the low-dimensional case $n = 2$.

In fact, the problem simplifies considerably when $n = 2$ due to the rigidity of low-dimensional settings.
Since the spacetime Weyl curvature $W$ now vanishes identically, it follows that any vacuum aAdS spacetime when $n = 2$ must be locally isometric to the AdS spacetime.
Furthermore, as all curvature terms disappear from the system \eqref{eq.transport_init}, one can prove unique continuation using only transport equations (and avoiding wave equations).
This yields analogues of all our main theorems for $n=2$---but \emph{from any domain $\mi{D}$, without requiring the GNCC}.

\subsection{Extensions of the GNCC}

An often studied setting in the physics literature is the case when the boundary region $\mi{D} \subset \mi{I}$ in Theorem \ref{thm.correspondence} is a causal diamond, that is, of the form
\begin{equation}
\label{eq.causal_diamond} \mi{D} := \mc{I}^+ (p) \cap \mc{I}^- (q) \text{,} \qquad p, q \in \mi{I} \text{.}
\end{equation}
($\mc{I}^+$ and $\mc{I}^-$ denote the chronological future and past, respectively, in $( \mi{I}, \gm )$.)
Unfortunately, causal diamonds \eqref{eq.causal_diamond}, regardless of how large they are, generically fail to satisfy the GNCC when $n > 2$; see the argument in \cite[Section 3.3]{chatz_shao:uc_ads_gauge}.
As a result, Theorem \ref{thm.correspondence} fails to apply when $\mi{D}$ is as in \eqref{eq.causal_diamond}---in other words, our result cannot establish that vacuum aAdS spacetimes are uniquely determined by their holographic data on any non-pathological causal diamond.

This leads to the question of whether the GNCC can be further refined, so that Theorem \ref{thm.correspondence} can be extended to apply to $\mi{D}$ as in \eqref{eq.causal_diamond} in some capacity.
One observation here is that the failure of the GNCC is due only to the presence of corners in $\partial \mi{D}$, where the boundaries of $\mc{I}^+ (p)$ and $\mc{I}^- (q)$ intersect.
Near these corners, one can find near-boundary null geodesics ``hovering over" but avoiding $\mi{D}$.
This leads to the following question:\ \emph{could holographic data on $\mi{D}$ uniquely determine the vacuum aAdS spacetime near some proper subset $\mi{D}' \subset \mi{D}$}---in particular when $\mi{D}$ is sufficiently large, and when $\mi{D}'$ is sufficiently far from any corners in $\partial \mc{I}^+ (p) \cap \partial \mc{I}^- (q)$?

At the same time, one may ask whether this refined GNCC can also be formulated for more general domains $\mi{D} \subset \mi{I}$.
More specifically, one can formulate the following:

\begin{problem} \label{prb.egncc}
Consider the setting of Theorem \ref{thm.correspondence}.
Show that if $\mi{D}' \subseteq \mi{D}$ satisfies some (yet to be formulated) ``refined GNCC" relative to $\mi{D}$, then $( \mi{M}, g )$ is uniquely determined (up to isometry) near $\mi{D}'$ by the holographic data on $\mi{D}$ (up to gauge equivalence).
\end{problem}

Keeping with our intuitions, an optimal formulation of such an extended GNCC would be one that directly characterizes null geodesic trajectories near the conformal boundary.
Such a criterion would fully confirm the belief that the only impediment to unique continuation for waves from the conformal boundary is the near-boundary geometric optics solutions of Alinhac and Baouendi \cite{alin_baou:non_unique}.
However, a proof of such a statement may require novel ideas from microlocal analysis.

\subsection{Global Correspondences}

Recall that our main result, Theorem \ref{thm.correspondence}, is ``local" in nature, in the sense that the vacuum spacetime is only uniquely determined near a given conformal boundary region $\mi{D}$.
This is due to our general setup, which does not provide any information on the global spacetime geometry.
However, this leaves open the question of \emph{whether a more global unique continuation result can be established if more additional assumptions are imposed}.

For example, one can consider aAdS spacetimes $( \mi{M}, g )$ that are global perturbations, in some sense, of a Kerr-AdS spacetime.
One can then ask \emph{whether any given holographic data $( \mi{I}, \gb{0}, \gb{n} )$ determines the spacetime in the full domain of outer communications, or if additional conditions are needed to rule out bifurcating counterexamples}.
Also, if uniqueness holds, then another question of interest is \emph{whether one has a one-to-one correspondence between holographic data on the boundary and some (appropriately conceived) holographic data on the black hole horizon}.

In \cite[Section 6]{hol_shao:uc_ads}, the Carleman estimates of that paper were used to show that the linearized EVE on AdS spacetime (formulated as Bianchi equations for spin-$2$ fields) are globally characterized by their holographic boundary data over a sufficiently long timespan.
Upcoming work by McGill and the second author will extend this to the nonlinear setting---roughly, under additional global topological assumptions, AdS spacetime is globally uniquely determined, as a solution to the EVE, by its holographic data.
An interesting next step would be to explore whether these analyses can be extended to black hole aAdS spacetimes, such as Schwarzschild-AdS and Kerr-AdS.

\subsection{Nonlinear Counterexamples}

As mentioned before, all the intuitions for the GNCC being the crucial criterion for Theorem \ref{thm.correspondence} are based on the observation that violations of the GNCC lead to counterexamples to unique continuation for wave equations.
Certainly, it would be more preferable to \emph{construct counterexamples to unique continuation for the nonlinear EVE itself when the GNCC fails to hold}.
However, this problem is wide open and is expected to be very difficult.

\subsection{Nature of the Correspondence}

While Theorem \ref{thm.correspondence} establishes that there exists a one-to-one correspondence between vacuum aAdS spacetimes (near the conformal boundary) and some space of holographic data, it says little about the nature of the correspondence, leaving open the question of how the gravitational dynamics and the conformal field theories are connected.

One can pose a number of interesting problems here that are both wide open and challenging.
For instance, one such problem would be to \emph{determine which holographic data $( \mi{I}, \gb{0}, \gb{n} )$ can be realized as the conformal boundary data for a vacuum aAdS spacetime}.
Moreover, given holographic data that is associated to a vacuum aAdS spacetime, another challenge would be to \emph{reconstruct or approximate the aAdS spacetime that induces the holographic data}.

\raggedright
\raggedbottom
\bibliographystyle{amsplain}
\bibliography{articles, books, misc}

\end{document}